\documentclass[a4paper,UKenglish,cleveref,autoref,thm-restate,
runningheads,orivec]{llncs}

\usepackage[dvipsnames]{xcolor}
\usepackage[disable]{todonotes}
\usepackage{paralist}
\usepackage{amssymb}
\usepackage[a4paper, margin=2.5cm]{geometry}
\usepackage{apxproof} 
\usepackage{pgf}
\usepackage{tikz}
\usetikzlibrary{arrows,shapes,automata,petri}

\DeclareMathAlphabet{\pazocal}{OMS}{zplmf}{m}{n}
\newcommand{\mcal}[1]{\pazocal{#1}}

\newcommand{\rulename}[1]{\ensuremath{\mbox{\textsc{#1}}}}

\newcommand{\qt}[1]{``{#1}"}

\newcommand{\abc}{\emph{AbC}\xspace}

\newlength{\arrow}
\newcounter{sqindex}

 \newcommand{\rom}[1]{ \textup{(\lowercase\expandafter{\romannumeral#1})}}

\usepackage[ruled]{algorithm2e} 

\SetAlFnt{\small}
\SetAlCapFnt{\small}
\SetAlCapNameFnt{\small}
\SetAlCapHSkip{0pt}
\IncMargin{-\parindent}

\newcommand \Until      {\mathbin{\mcal{U}\kern-.1em}}
\newcommand \Release     {\mathbin{\mcal{R}\kern-.1em}}

\newcommand \Since      {\mathbin{\mcal{S}\kern-.08em}}

\newcommand \g    {\mathsf{{G}\kern.08em}}
\newcommand \f    {\mathsf{{F}\kern.08em}}
\newcommand \UntilHat   {\mathbin{\LTLhat{\mcal{U}}\kern-.1em}}

\newcommand \SinceHat   {\mathbin{\LTLhat{\mcal{S}}\kern-.08em}}

\renewcommand \phi      {\varphi}

\newcommand \ltal       {\textsc{ltol}\xspace}

\newcommand{\set}[1]{\{{#1}\}}
\newcommand{\conf}[1]{\langle{#1}\rangle}

\newcommand{\toall}{\star}

\def\<#1>{\mathinner{\langle#1\rangle}}

\newcommand{\msf}[1]{\mathsf{#1}}

\newcommand{\rcp}{{{\rulename{ReCiPe}}}\xspace}

\newcommand{\comment}[1]{}

\newcommand{\typecvar}{{\scriptstyle@\msf{type}}}
\newcommand{\assigncvar}{{\scriptstyle@\msf{asgn}}}
\newcommand{\readycvar}{{\scriptstyle@\msf{rdy}}}
\newcommand{\lnkcvar}{{\scriptstyle@\msf{lnk}}}

\newcommand{\np}[1]{\todo[author=NP, color=green!40]{#1}}

\newcommand{\listen}{\rulename{ls}}

\newcommand{\succi}[1]{\rightarrow_{\hspace*{-2pt}_{\leq_{#1}}}}
\newcommand{\comrel}{\rightarrow_c}
\newcommand{\comreli}[1]{\rightarrow_{c_{#1}}}
\newcommand{\intrel}{\rightarrow_i}
\newcommand{\hist}{\textbf{hist}}
\newcommand{\comp}{\textbf{comp}}

\newcommand{\leadsTo}[1]{\ensuremath{{\bf |}{#1}{\bf \rangle}}}
\newcommand{\pre}[1]{{\mkern-1mu ^\bullet {#1}}\xspace}
\newcommand{\inhpre}[1]{{\mkern-1mu ^\circ {#1}}\xspace}
\newcommand{\preT}{\pre{t}}
\newcommand{\inhpreT}{\inhpre{t}}
\newcommand{\post}[1]{{{#1}^\bullet}\xspace}

\newcommand{\postT}{\post{t}}

\newcommand{\Act}{\mathsf{Act}}

\newcommand{\lpo}{\ensuremath{\mathsf{LPO}}\xspace}
\newcommand{\glpo}{\ensuremath{\mbox{g-}\mathsf{LPO}}\xspace}
\newcommand{\lpos}{{{\lpo}s}\xspace}
\newcommand{\glpos}{{{\glpo}s}\xspace}

\newcommand{\incomparable}{\#}

\newcommand{\Amc}{\ensuremath{\mathcal{A}}\xspace}

\newcommand*{\glue}{%
  \multimap
}

\newif\iflisten\listentrue

\newtheoremrep{theorem}{Theorem}
\newtheoremrep{lemma}{Lemma}

\begin{document}
\title{Interleaving \& Reconfigurable Interaction: Separating Choice from
Scheduling using Glue\thanks{This work is funded by
the ERC consolidator grant
D-SynMA (No. 772459) and the Swedish research council
grants: SynTM (No.
2020-03401) and VR project (No. 2020-04963).}
}
\titlerunning{Separating Choice from Scheduling using Glue}
%
\author{
	Yehia Abd
	Alrahman\inst{1}
	\and
	Mauricio
	Martel\inst{1}
	\and
	Nir Piterman\inst{1}
}
\authorrunning{Y. Abd Alrahman et al.}
%
\institute{University of Gothenburg, Gothenburg, Sweden\\
	\email{\{yehia.abd.alrahman,nir.piterman,mauricio.martel\}@gu.se}}
\maketitle              

\begin{abstract}
Reconfigurable interaction induces another dimension of nondeterminism
in concurrent systems which makes it hard to reason about the different
choices of the system from a global perspective. Namely, (1) choices
that correspond to concurrent execution of independent events; and (2)
forced interleaving (or scheduling) due to reconfiguration.
Unlike linear order semantics of computations, partial
order semantics recovers information about the
interdependence among the different events for fixed interaction,
but still is unable to handle reconfiguration.
We introduce \emph{glued partial orders} as a way to capture
reconfiguration. Much like partial orders capture all possible choices
for fixed systems, glued partial orders capture all possible choices
alongside reconfiguration.
We show that a glued partial order is sufficient to correctly capture
all partial order computations that differ in forced interleaving
due to reconfiguration. Furthermore, we show that
computations belonging to different glued partial orders are only
different due to non-determinism.

\end{abstract}

\section{Introduction}
Reconfigurable concurrent
systems~\cite{info,DBLP:conf/atal/AlrahmanPP20,scp} are a class of
computational systems, consisting of a collection of processes (or
agents) that interact and exchange information in nontrivial
ways. Agents interact using message-passing~\cite{pi1} (or
token-passing~\cite{pnarc}) and based on dynamic notions of
connectivity where agents may only observe, inhibit or participate in
interactions happening on links they are connected to. Agents may get
connected or disconnected to links as side-effects of the interaction,
and thus providing dynamic and sophisticated scoping mechanisms of
interaction through reconfigurable interfaces.

Reconfiguration induces another dimension of nondeterminism in
concurrent systems where it becomes hard to reason about the different
choices of the system from a global perspective. It creates a
situation where some events must be ordered with respect to
\emph{sequences of other events} dynamically during execution, and
thus forcing interleaving in a non-trivial way. That is, from the
point of an event, a sequence of other events is considered as a
single block and can only happen before or after it.
Note that reconfiguration is an internal event, and is totally hidden
from the perspective of an external observer~\cite{observ} who may
only observe message-/token- passing.
Indeed, messages or tokens can only indicate the occurrence of
exchange but cannot help with noticing that a reconfiguration has
happened and what are the consequences of reconfiguration.
Knowing the reason why some event is scheduled before some others and
the causal dependencies among the different events is crucial to
facilitate reasoning about specific internal aspects from a global
perspective~\cite{BestD90}.
It also becomes very relevant when applying correct-by-construction
techniques~\cite{synth} to synthesise such systems.

Clearly, linear order semantics of computations~\cite{BaetenB01,Vogler02a}
cannot be used to globally distinguish a system choice due to
concurrent execution of independent events and a forced interleaving
due to reconfiguration. It cannot be even used to recover information
about the participants of an event and the interdependence of the
different events. Therefore, a partial order semantics of computations
is in-order. Existing approaches to partial order semantics
(cf. \emph{Process semantics} of Petri nets~\cite{PetriR08,meseguer1992semantics,Vogler02a}
and \emph{Mazurkiewicz traces} of Zielonka
automata~\cite{Zielonka87,GenestGMW10,KrishnaM13}) proved useful in recovering
information about  the participants of events and independence of
concurrent events. For instance, in the Process semantics of Petri nets,
two concurrent events can be executed in any order or even
simultaneously, and thus we can distinguish concurrent execution from
mere nondeterminism. However, these formalisms have fixed interaction
structures that define interdependence of events in a static way, and
thus leads to a straightforward partial order semantics. Indeed, while
the interdependence of events is statically defined based on the
structure of a Petri net, it is also defined based on the domains of
events of Zielonka automata which are fixed in advance.

In this paper, we propose a partial order semantics of computations
under reconfiguration. In such settings, dependencies among events
emerge dynamically as side-effects of interaction, and thus we handle
these emergencies while ensuring that the semantics defines the actual
behaviour of the system. Our approach consists of characterising
reconfiguration points and their corresponding scheduling decisions in
a single structure, while preserving a \emph{true}-concurrent
execution of independent events. Our semantics allows reasoning about
the individual behaviour of agents composing the system and their
interaction information.
We test our results on \emph{Petri net with inhibitor
arcs} (PTI-nets)~\cite{FlynnA73,pnarc} and \emph{Channeled Transition Systems}
(CTS)~\cite{alrahman2021modelling,DBLP:conf/atal/AlrahmanPP20}.
These modelling frameworks
cover a wide range of interaction capabilities alongside
reconfiguration from two different schools of concurrency. In fact,
inhibitor arcs add a restricted form of reconfiguration to Petri nets
while CTS can be considered as a generalisation of Zielonka automata,
supporting rich interactions alongside reconfiguration.

\noindent
{\bf Contributions.} We define  specialised partial orders, that we
call \emph{labelled partial orders} ($\mathsf{LPO}$ for short), to
represent computations. An $\mathsf{LPO}$ is a representation of a
specific computation. That is, given a system consisting of a set of
agents, we can construct an $\mathsf{LPO}$ by only considering the
local views of individual agents and their interaction information. An
$\mathsf{LPO}$ defines how the individual computations of agents are
related, and also how different events are related.
In the spirit of Mazurkiewicz traces, the states of
different agents are (strictly) incomparable, that is there is no
notion of a global state.  This way we can easily single out finite
sequences of computation steps where an agent or a (small) group of
agents execute independently. We can also distinguish individual
events from joint ones. Despite the fact that an $\mathsf{LPO}$ may
refer to reconfiguration points, it cannot fully characterise
reconfiguration in a single structure. For this reason, we introduce
\emph{glued labeled partial orders} (g-$\mathsf{LPO}$, for short),
that is an extension of $\mathsf{LPO}$ with \emph{glue} to separate a
non-deterministic choice from forced scheduling due to
reconfiguration. Intuitively, two elements are glued from the point of
view of another element if they both happen either before or after
said element.
We show that a g-$\mathsf{LPO}$ is sufficient to represent
$\mathsf{LPO}$ computations that  differ in scheduling due to
reconfiguration. We also show that $\mathsf{LPO}$ computations
belonging to different g-$\mathsf{LPO}$(s) are different due to
nondeterministic selection of independent events.

The paper is organised as follows:
In Sect.~\ref{sec:informal}, we informally
present our partial order semantics and in
Sect.~\ref{sec:background}, we introduce the necessary background.
In Sect.~\ref{sec:sem}, we provide $\mathsf{LPO}$ semantics for PTI-nets
 and CTSs.
In Sect.~\ref{sec:glue} we define glued partial orders and the
corresponding extension to both PTI-nets and CTSs. We show, for
both, that every $\mathsf{LPO}$ computation is only a refinement of some
g-$\mathsf{LPO}$ of the same system.
In Sect.~\ref{sec:separation} we prove important results on
g-$\mathsf{LPO}$ with respect to reconfiguration and
nondeterminism. In Sect.~\ref{sec:conc} we present concluding remarks,
related works, and future directions.
All proofs are included in the appendix.

\section{Labelled Partial Order Computations in a  Nutshell}
\label{sec:informal}
In this section, we use a fragment of a PTI-net to informally illustrate the \lpo semantics under reconfiguration
and the idea behind \glpo.

We consider the PTI-net in Fig.~\ref{fig:petriex}(a), where we
interpret reconfiguration and concurrency in the following way: each
token represents an individual agent and the structure of the net
defines the combined behaviour.
The places of the net, denoted by circles, define the states of the
different agents  during execution.
The transitions, denoted by squares, can either refer to synchronisation
points (e.g., $t_1$ and $t_2$) or individual computation steps (e.g.,
$t_3$ and $t_4$).

Arrows define which places require to have tokens to enable a transition
and the places to put tokens after firing.
In our examples all arrows consume/produce one token.
For instance, transition $t_1$ may fire when
there is at least one token in both $p_1$ and $p_2$.
Transition firing induces removal of tokens
from input places and addition of tokens in output places.
Thus, when $t_1$ fires, one token is removed from $p_1$ and one from
$p_2$ and one token is placed in $p_3$ and $p_4$, each.
Sometimes a place can choose nondeterministically which transition
to participate in (e.g., $p_4$ chooses $t_2$ or $t_3$).
A place can inhibit the firing of some transition (e.g.,
$p_3$ inhibiting $t_4$) using an inhibitor arc ($p_3 \glue t_4$).
While the place contains a token it inhibits the transition.
We interpret this as the agent represented by the token (e.g., in $p_3$)
starting to listen  to the transition ($t_4$), but it cannot
participate, and thus it inhibits its execution.
In our example, in $p_1$ the agent is not listening to $t_4$, but once
$t_1$ is executed the agent reconfigures its interaction interface and
starts listening.
This means that $t_4$ may only fire either before a token is placed in
$p_3$ or after the token is removed.
Clearly, this can only happen when $t_4$ either happens before $t_1$ or
after $t_2$.
Thus from the point of view of $t_4$ both
$t_1$ and $t_2$ are considered as a single block, and their execution
cannot be interrupted.
Namely, the only viable sequences of execution (in case $t_2$ is
scheduled later) are $t_4,{\color{red}t_1,t_2}$ or
${\color{red}t_1,t_2},t_4$.
Note that this is only from the point of view of $t_4$ and has no
implications for other transitions.
Indeed, other transitions can have a different point of view (e.g.,
$t_3$).
This creates a forced interleaving in a non-trivial way due to the
occurrence of non-observable events (i.e., reconfiguration) that we
cannot reason about from a global perspective.
Furthermore, these dependencies among events emerge dynamically
as side-effects of interaction, and thus put the correctness of
partial order semantics at stake.

To handle this issue, we introduce a partial order semantics of
computations under reconfiguration.
We handle the above mentioned emergences by characterising
reconfiguration points and their corresponding scheduling
decisions in a single structure, while preserving a
\emph{true}-concurrent execution of independent events.
Our semantics allows reasoning about the individual behaviour of agents
composing the system and their interaction information.

We illustrate our $\mathsf{LPO}$ and g-$\mathsf{LPO}$ semantics in
Fig.~\ref{fig:petriex}(b), which characterises all possible (maximal)
computations of the net.
Here, we use the arrow $\rightarrow$ to indicate a happen before relation.

The two figures succinctly encode \emph{three} possible \lpos:
\begin{inparaenum}[(i)]
\item
	the \lpo obtained from Fig.~\ref{fig:petriex}(b) left structure
	with the dashed arrow from $t_4$ to $t_1$;
\item
	the \lpo obtained from Fig.~\ref{fig:petriex}(b) left structure
	with the dashed arrow from $t_2$ to $t_4$; and
\item
	the \lpo obtained from Fig.~\ref{fig:petriex}(b)
	right structure with the dashed arrow from $t_4$ to $t_1$.
\end{inparaenum}
\lpos (i) and (ii) agree that the token in $p_4$ nondeterministically
chooses the transition $t_2$ while in (iii) the nondeterministic choice
is $t_3$.
All \lpos capture information about interaction and interdependence
among events.
Indeed, in all cases we see that both $p_1$ and $p_2$ synchronize
through the transition $t_1$.
Places that are not strictly ordered with respect to a common
transition are considered concurrent.
Thus, as in Mazurkiewicz traces there is no notion of a global state.
Notice that \lpos (i) and (ii) differ only in the forced interleaving
of $t_4$ with respect to the block $t_1,t_2$.

Notice that both \lpos (i) and (ii) have information both on
reconfiguration and nondeterminism, but each individually cannot be
used to distinguish the hidden reconfiguration.
In fact, $t_4\rightarrow t_1$ in (i) indicates that $t_4$ happened
before a reconfiguration caused by $t_1$, and $t_2\rightarrow t_4$ in (ii)
indicates that $t_4$ happened after the reconfiguration.
In (iii), due to the different nondeterminsitic choice, the only
possible case we have to consider is that of $t_4$ happening before
$t_1$.


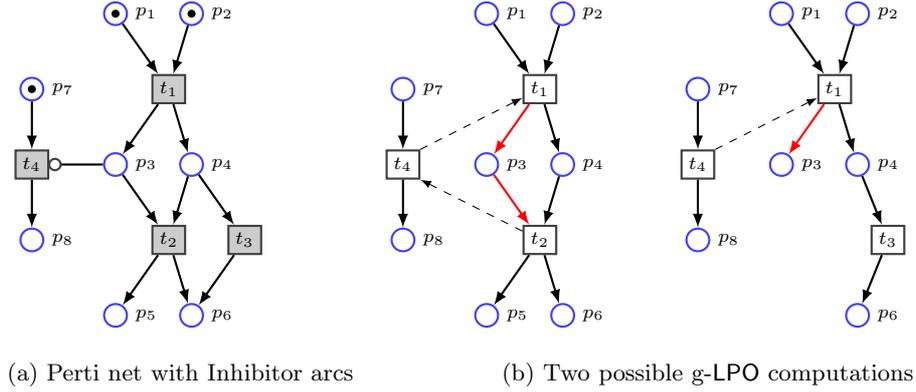
\begin{figure}[t]
	\centering
	\begin{tabular}{c@{\qquad\qquad}c@{\qquad\qquad}c}
	\scriptsize
	\begin{tikzpicture}[node distance=1cm]
		\tikzstyle{place}=[circle,thick,draw=blue!75,fill=white!20,minimum
		size=3mm]
		\tikzstyle{transition}=[rectangle,thick,draw=black!75,
		fill=black!20,minimum size=1mm]
		\tikzstyle{inhibitor}=[circle,thick,scale=0.2mm,draw=black!75,
		fill=white]

		\node[place,tokens=1,label=right:$p_1$]
			(p1)  {};
		\node[place,tokens=1,label=right:$p_2$]
			(p2) [right of=p1] {};

	 	\node[transition]
			(t1) [below of=p1,xshift=7mm] {$t_1$};

		\node[place,label=right:$p_3$]
			(p3) [below of=t1,xshift=-7mm] {};
		\node[place,label=right:$p_4$]
			(p4) [right of=p3] {};

		\node[transition]
			(t2) [below of=p3,xshift=7mm]  {$t_2$};
		\node[transition]
			(t3) [right of=t2]  {$t_3$};

		\node[place,label=right:$p_5$]
			(p5) [below of=t2,xshift=-7mm] {};
		\node[place,label=right:$p_6$]
			(p6)  [right of=p5]  {};
		\node[place,tokens=1,label=right:$p_7$]
			(p7) [left of=t1,xshift=-.8cm] {};
			\node[transition]
			(t4) [left of=p3,xshift=-.1cm]  {$t_4$};
		\node[place,label=right:$p_8$]
			(p8) [below of=t4]  {};

		\node[inhibitor]
			(inhT) [left of=p3,xshift=-.4cm] {};

		\draw[-latex,thick] (p1) -- node[above] {} (t1);
		\draw[-latex,thick] (p2) -- (t1);
		\draw[-latex,thick] (t1) -- (p3);
		\draw[-latex,thick] (t1) -- node[above] {} (p4);
		\draw[-latex,thick] (p3) -- (t2);
		\draw[-latex,thick] (p4) -- node[above] {} (t2);
		\draw[-latex,thick] (t2) -- node[above] {} (p5);
		\draw[-latex,thick] (t2) -- (p6);
		\draw[-latex,thick] (p7) -- (t4);
		\draw[-latex,thick] (t4) -- (p8);
		\draw[-latex,thick] (p4) -- (t3);
		\draw[-latex,thick] (t3) -- (p6);

		\draw[thick] (inhT) -- (p3);

	\end{tikzpicture} &\quad
	\scriptsize
	\begin{tikzpicture}[node distance=1cm]
		\tikzstyle{place}=[circle,thick,draw=blue!75,fill=white!20,minimum
		size=3mm]
		\tikzstyle{transition}=[rectangle,thick,draw=black!75,
		fill=white,minimum size=1mm]
		\tikzstyle{inhibitor}=[circle,thick,scale=0.2mm,draw=black!75,
		fill=white]

		\node[place,label=right:$p_1$]
			(p1)  {};
		\node[place,label=right:$p_2$]
			(p2) [right of=p1] {};

	 	\node[transition]
			(t1) [below of=p1,xshift=7mm] {$t_1$};

		\node[place,label=right:$p_3$]
			(p3) [below of=t1,xshift=-7mm] {};
		\node[place,label=right:$p_4$]
			(p4) [right of=p3] {};

		\node[transition]
			(t2) [below of=p3,xshift=7mm]  {$t_2$};

		\node[place,label=right:$p_5$]
			(p5) [below of=t2,xshift=-7mm] {};
		\node[place,label=right:$p_6$]
			(p6)  [right of=p5]  {};
		\node[place,label=right:$p_7$]
			(p7) [left of=t1,xshift=-.8cm] {};
			\node[transition]
			(t3) [left of=p3,xshift=-.1cm]  {$t_4$};
		\node[place,label=right:$p_8$]
			(p8) [below of=t3]  {};

		\draw[-latex,thick] (p1) -- node[above] {} (t1);
		\draw[-latex,thick] (p2) -- (t1);
		\draw[-latex,thick,color=red] (t1) -- (p3);
		\draw[-latex,thick] (t1) -- node[above] {} (p4);
		\draw[-latex,thick,color=red] (p3) -- (t2);
		\draw[-latex,thick] (p4) -- node[above] {} (t2);
		\draw[-latex,thick] (t2) -- node[above] {} (p5);
		\draw[-latex,thick] (t2) -- (p6);
		\draw[-latex,thick] (p7) -- (t3);
		\draw[-latex,thick] (t3) -- (p8);
		\draw[-latex,dashed] (t3.north east) -- (t1);
	\draw[-latex,dashed] (t2) -- (t3.south east);
	\end{tikzpicture} &\hspace{-5mm}
	\scriptsize
	\begin{tikzpicture}[node distance=1cm]
		\tikzstyle{place}=[circle,thick,draw=blue!75,fill=white!20,minimum
		size=3mm]
		\tikzstyle{transition}=[rectangle,thick,draw=black!75,
		fill=white,minimum size=1mm]
		\tikzstyle{inhibitor}=[circle,thick,scale=0.2mm,draw=black!75,
		fill=white]

		\node[place,label=right:$p_1$]
			(p1)  {};
		\node[place,label=right:$p_2$]
			(p2) [right of=p1] {};

	 	\node[transition]
			(t1) [below of=p1,xshift=7mm] {$t_1$};

		\node[place,label=right:$p_3$]
			(p3) [below of=t1,xshift=-7mm] {};
		\node[place,label=right:$p_4$]
			(p4) [right of=p3] {};

		\node[transition]
			(t2) [below of=p3,xshift=14mm]  {$t_3$};

		\node[place,label=right:$p_6$]
			(p6)  [right of=p5]  {};
		\node[place,label=right:$p_7$]
			(p7) [left of=t1,xshift=-.8cm] {};
			\node[transition]
			(t3) [left of=p3,xshift=-.1cm]  {$t_4$};
		\node[place,label=right:$p_8$]
			(p8) [below of=t3]  {};

		\draw[-latex,thick] (p1) -- node[above] {} (t1);
		\draw[-latex,thick] (p2) -- (t1);
		\draw[-latex,thick,color=red] (t1) -- (p3);
		\draw[-latex,thick] (t1) -- node[above] {} (p4);
		\draw[-latex,thick] (p4) -- node[above] {} (t2);
		\draw[-latex,thick] (t2) -- (p6);
		\draw[-latex,thick] (p7) -- (t3);
		\draw[-latex,thick] (t3) -- (p8);
		\draw[-latex,dashed] (t3.north east) -- (t1);
	\end{tikzpicture}

	\end{tabular}\\[2ex]
	\begin{tabular}{c@{\qquad\quad}c@{\qquad\quad}c}
(a) Perti net with Inhibitor arcs && (b) Two possible g-$\mathsf{LPO}$
computations
	\end{tabular}
	\caption{Petri net with inhibitor arcs}
	\label{fig:petriex}
\end{figure}

This suggests that we can actually isolate reconfiguration from
nondeterminism by using a more sophisticated structure than
$\mathsf{LPO}$, and thus expose the difference in a way that allows
reasoning about these hidden events from a global perspective. For this
reason, we define g-$\mathsf{LPO}$ computations, that are an extension
of $\mathsf{LPO}$ with a notion of \emph{glue}.

For PTI-nets like fixed systems, a g-$\mathsf{LPO}$  
simply drops strict ordering of events
with respect to each other (like $t_4\rightarrow t_1$ or $t_2\rightarrow t_4$), and
instead assigns each event a (possibly empty) glue relation defining
the \emph{glued} elements from the point of view of that event.
The glue relation is defined based on reconfiguration points, and in
case of Petri nets is based on inhibitor arcs.
We will see later how this is defined in a more dynamic and
compositional model like CTS, where structural information does not
simply exist. There, the g-$\mathsf{LPO}$  has to account also to
event-to-event ordering when sharing the same communication channel.

Consider now the structures in Fig.~\ref{fig:petriex}(b) without the
dashed arrows and, now, with an explanation of the red arrows.
These two structures are each a g-$\mathsf{LPO}$.
For the one on the left, since $p_3$ inhibits $t_4$ all existing
incoming and outgoing edges from $p_3$ are glued to $p_3$.
Thus, $t_4$'s glue relation includes these edges (in red).
All other transitions have empty glue relations because they are not
inhibited.
As they are not inhibited, their interdependence is well-captured
statically based on the structure of the net.
Note that the glue relation is not required to be transitive and the
glue only relates places and transitions.
In the structure on the right of the figure, $t_1$ is glued only to
$p_4$.
As $t_3$ is scheduled rather than $t_2$, then $p_3$ remains as a maximal
element.

As we show later, a single g-$\mathsf{LPO}$ can be used to characterise
reconfiguration and separate it from other sources of nondeterminism in
the system.
\section{Preliminaries: Labeled Partial Orders}
\label{sec:background}

\label{subsec:po and lpo}

We use partial orders to represent computations. We specialize notations
to match our needs.

A \emph{partial order} ($\mathsf{PO}$, for short) is a binary relation
$\leq$ over a set $O$ that is
reflexive, antisymmetric, and transitive.
We use $a<b$ for $a\leq b$ and $a\neq b$.
We use $a\incomparable b$ for $a\not\leq b$ and $b\not\leq a$, i.e., $a$
and $b$ are incomparable.

%

A \emph{labelled partial order} ($\mathsf{LPO}$, for short) is
$(O,\rightarrow_c,\rightarrow_i,\Sigma,\Upsilon,L)$,  where
$O=V\biguplus E$
is a set of elements partitioned to nodes and edges, respectively,
$\rightarrow_c$ and $\rightarrow_i$ are disjoint, anti-reflexive,
anti-symmetric, and non-transitive \emph{communication} and \emph{interleaving} order
relations over $O$. We have $\rightarrow_c\subseteq V\times E \cup E
\times V$ and $\rightarrow_i \subseteq E\times E$.
When $\rightarrow_i=\emptyset$ we omit it from the tuple.
The relation $\leq$ is the reflexive and transitive closure of the
union of $\rightarrow_c$ and $\rightarrow_i$.
We require that $\leq$ is a partial order.
Moreover, $\Sigma$ is a node alphabet, $\Upsilon$ is an edge
alphabet, and $L:O\rightarrow \Sigma\cup \Upsilon$ such that
$L(V)\subseteq \Sigma$ and $L(E)\subseteq \Upsilon$   is the labelling function.

Intuitively, elements in $V$ can denote states or execution histories of
individual agents and elements in $E$ denote transitions or events.
Thus, a history belongs to an individual agent and a
transition corresponds to either an individual computational step or
a synchronisation point among multiple agents.
The relation $\rightarrow_c$ captures participation in communication
and the relation $\rightarrow_i$ captures order requirements.

We denote $\rightarrow = \rightarrow_c \cup \rightarrow_i$.
Given an element $a\in O$ we write $\pre{a}$ for $\{ b ~|~
b\rightarrow a \}$ and $\post{a}$ for $\{b ~|~ a\rightarrow b\}$.
\section{LPO Semantics}
\label{sec:sem}
In this section, we present Petri Nets with inhibitor
arcs~\cite{FlynnA73,pnarc}
and Channeled Transition
Systems~\cite{alrahman2021modelling,DBLP:conf/atal/AlrahmanPP20} and we
provide each with a labelled partial order semantics.
The labelled partial order semantics of Petri nets extends occurrence
nets~\cite{meseguer1992semantics} with event-to-event connections that
allow to capture reconfigurations.
We include in appendix the labelled partial order semantics of asynchronous automata, which do
not require the relation $\intrel$, and, thus, show that the separation
of results in this paper only make sense in reconfigurable systems.

\begin{toappendix}
\subsection{Asynchronous automata}
\label{subsec:async aut}

We include the \lpo semantics of asynchronous automata.
As asynchronous automata do not have reconfigurations of
communication we only need the communication relation and do not use
the interleaving order relation.
This makes the notion of glue not relevant for asynchronous automata.

A \emph{process} is $P=(\mathsf{Act},S,s^0,\delta)$, where $\Act$ is a
finite
non-empty alphabet, $S$ is a finite and non-empty set of states,
$s^0\in S$ is an initial state, and $\delta \subseteq S \times \Act
\times S$.
We also write $\delta:S\times \Act\rightarrow 2^S$ when convenient.

A \emph{history} $h=s_0,\ldots, s_n$ is a finite sequence such that
$s_0=s^0$ and for
every $0\leq i <n$ we have $s_{i+1} \in \delta(s_i,a_i)$ for some
$a_i\in \Act$. The length of $h$ is $n+1$, denoted $|r|$.
For convenience, if $h_1=s_0,\ldots, s_n$ and $h_2=s_0,\ldots,
s_n,s_{n+1}$ such that $s_{n+1}\in \delta(s_n,a_n)$ we write $h_2\in
\delta(h_1,a_n)$ or $\delta(h_1,a_n,h_2)$.
We define $\hist(P)$ to be the set of histories of $P$.

A finite asynchronous automaton \Amc with $n$ processes is $\Amc =
(P_1,\ldots, P_n)$ such that each $P_i = (\Act_i, S_i,
s^0_i,\delta_i)$ is a process.
Let $\Act= \bigcup_{i}\Act_i$.

\begin{definition}[computation]
	\label{def:computation async aut}
A \emph{computation} of \Amc is an $\mathsf{LPO}$ $(O,\comrel,
\Sigma,\Upsilon,L)$, where $V \subseteq \bigcup_{i}\hist(P_i)$,
$\Sigma = V$, $L(h) = h$, and $\Upsilon=\Act$ such that:
\begin{enumerate}
\item
  The edge $e_\epsilon$ such that $L(e_\epsilon)=a$ for some $a$ is the
  unique minimal element according to $\leq$.
  For every $i$, we have $s^0_i \in V$ and $e_\epsilon\comrel s^0_i$.
\item
  If $h\in V$ there is a unique $e\in E$ such that $e\comrel h$. If
  $|h|>1$, there is also a unique $h'\in V$ such that  $h\in
  \delta(h',L(e))$ and $h'\comrel e$.
  \item
  For every $h\in V$ there is at most one $e\in E$ such that
   $h\comrel e$.
\item
  For every $e\in E$ there is $I \subseteq [n]$ such that all the
following hold:
\begin{enumerate}
\item
 $L_{E}(e) \in \bigcap_{i \in I}  \Act_i\setminus \bigcup_{i \notin I}
  \Act_i$
\item
  $\pre{e},\post{e} \in \bigcup_{i \in I} \hist(P_i)$
\item
  For every $i\in I$ we have $|\pre{e} \cap \hist(P_i)|=1$ and
  $|\post{e}\cap \hist(P_i)|=1$
\end{enumerate}
\end{enumerate}
\end{definition}
That is, a computation starts from an arbitrary joint edge that leads
to the initial states of all processes.
For every transition, the set of participating processes is all
those having the transition's label in their alphabet.
Each participating process has a history that is a predecessor of the
transition and a history that is a successor of the transition.
The pair of histories that belong to one process satisfy the transition
relation of that process.






\end{toappendix}

\subsection{Petri Nets with Inhibitor Arcs (PTI-nets)}
\label{subsec:petri nets}
A Petri net $N$ with inhibitor arcs is a bipartite directed graph $N =
\conf{P, T, F, I}$, where $P$ and $T$ are the set of places
and transitions such that $P \cap T = \emptyset$, $F:(P \times
T) \cup (T \times P) \rightarrow \mathbb{N}$ is the flow relation, and
$I\subseteq (P\times T)$ is the inhibiting relation.
We write $(s,s')\in F$ for $F(s,s')>0$.
{We restrict attention to Petri nets where all transitions have a
non-empty preset.}

\np{This paragraph can be rephrased to more general vector notation for
PNs rather than talking about states.}
The configuration of a Petri net at a time instant is
defined by means of a \emph{marking}.
Formally, let $N$ be a Petri net with a set of places
$P=\set{p_1,\dots,p_k}$.
A marking is a function $m: P \rightarrow \mathbb{N}$ and is defined as
a vector $\overline{m}=m[1],\dots,m[k]$ where $m[i]$ corresponds to the
number of tokens in $p_i$, for $i=1,\dots,k$.
Vectors can be added, subtracted, and compared in the usual way.
We assume some initial marking $m_0$.
For $p \in P$ let $\vec{p}$ be the singleton vector
$\vec{p}:P\rightarrow
\{0,1\}$ such that $\vec{p}(p)=1$ and $\vec{p}(p')=0$ for every
$p'\neq p$.

For a transition $t\in T$ we define the \emph{pre-vector} of $t$,
denoted by $\preT$, to represent the vector ${\preT}[1],
\ldots,{\preT}[k]$, where $\preT[i]=F(p_i,t)$.
Similarly, the \emph{post-vector} of $t$ is $\postT={\postT}[1],\ldots,
{\postT}[k]$, where $\postT[i]=F(t,p_i)$.

%
%
%
%
%

An inhibitor arc from a place to a transition means that the transition
can only fire if no token is on that place.
The inhibitor set of a transition t is the set $\inhpreT = \set{p \in P
	\mid (p, t) \in I}$, and represents the places to be ``tested for
absence'' of tokens.
That is, an inhibiting place allows to prevent the transition firing.

\np{next two paragraphs can be removed. Perhaps replaced by a
reference.}
A transition $t$ is enabled at $m$ if for every $p \in \preT$ we have
$m(p) \geq F(p,t)$ and all inhibitor places are empty, i.e., for every
$p \in \inhpreT$ we have $m(p) = 0$.
Note that if for some $t$ and $p\in\inhpreT$ we have $(p,t) \in F$ then
$t$ can never fire, thus it is called blocked.

A transition $t$ enabled at marking $m$ can fire and produce a new
marking $m'$ such that $m' = m - \preT + \postT$, denoted
$m\leadsTo{t}m'$.
That is, for every place $p \in P$, the firing transition $t$ consumes
$F(p,t)$ tokens and produces $F(t,p)$ tokens.

\begin{definition}[History]
	\label{def:pn history}
	We define the set of histories of a net $N$ by induction.

	We define a special transition $t_\epsilon$ such that
	$\post{t_\epsilon}=m_0$. The pair $(\emptyset,t_\epsilon)$ is a
	t-history. Note that $t_\epsilon$ is not a transition in $T$.

%
	For a place $p$, let $h=(S,t)$ be a t-history
	such that $\post{t}(p)>0$.
	Then we have $(h,p,\post{t}(p))$ is a p-history.
	That is, given a t-history $h$ ending in transition $t$, where $p$ is
	in $\post{t}$, then the combination of $h$, $p$, and the number of
	tokes that $t$ puts in $p$ form a p-history.

	Consider a transition $t\in T$.
	A t-history is a pair
	$(S,t)$, where $S=\{(h_1,i_1),\allowbreak\ldots, (h_n,i_n)\}$ is a multiset
	satisfying the following.
	For every $j$ we have $h_j=(-,p,c_j)$ is a p-history, where
	$c_j\geq i_j$ and $\pre{t} = \sum_j i_j \cdot \vec{p_j}$.
	That is, the t-history identifies the set of p-histories from which
	$t$ takes tokens with the multiplicity of tokens taken from every
	p-history.

	Let $\hist(N)$ be the set of all histories of $N$ partitioned to
	$\hist_p(N)$ and $\hist_t(N)$ in the obvious way.
	Given a t-history $h=(S,t)$ and a p-history $h'$ we write
	$h(h')$ for the number of appearances of $h'$ in the multiset $S$.
\end{definition}

Now, everything is in place to define the labelled partial order
semantics of a PTI-net.

\begin{definition}[\rulename{LPO}-computation]
	\label{def:lpo-computation}
	A \emph{computation} of $N$ is an $\mathsf{LPO}$
	$(O,\comrel,\intrel,\Sigma,\Upsilon,L)$, where
	$V \subseteq \hist_p(N)$, $E \subseteq \hist_t(N)$,
	$\Sigma=P$, $\Upsilon=T$, for a p-history $v=(-,p,i)$
	we have $L(v)=p$ and for a t-history $(S,t)$
	we have $L(e)=t$, and such that:
	\begin{enumerate}
		\item[N1.]
		The t-history $(\emptyset,t_\epsilon)$ is the unique minimal
		element according to $\leq$.
 		\item[N2.]
		For a p-history $v=(e,p,i)\in V$ we have $e\in E$ and $e$ is
		the unique edge such that $e\comrel v$.
		\item[N3.]
		For a p-history $v=(h,p,i)\in V$, let $e_1,\ldots, e_j$ be the
		t-histories such that $v\comrel e_j$.
		Then, for every $j$ we have $e_j(v)>0$ and $\sum_j e_j(v)\leq
		i$.
		That is, $v$ leads to t-histories that contain it with
		the multiplicity of $v$ being respected.
		\item[N4.]
		For every $e\in E$, where $e=(\{(v_1,i_1),\ldots, (v_n,i_n)\},t)$,
		all the following hold:
		\begin{enumerate}
			\item[(a)]
			$\pre{e}\cap V =\{v_1,\ldots, v_n\}$ and
			$\post{e}\cap V=\{(e,p,\post{t}(p)) ~|~ \post{t}(p)>0\}$.

			\item[(b)]
			For every $v\in V$ such that $L(v) \in \inhpre{L(e)}$
			we have $e\leq v$ or $v\leq e$.
			\item[(c)]
			If $e\intrel e'$ then there is some $v$ such that
			either
			(i) $v\comrel e$ and $(L(v),L(e'))\in I$
			or
			(ii) $e'\comrel v$ and $(L(v),L(e))\in I$.
		\end{enumerate}
	\end{enumerate}
\end{definition}

That is, a computation starts from the dummy transition $t_\epsilon$,
which establishes the initial marking.
Every other transition is a t-history that connects the p-histories
that it contains.
If a place inhibits a transition then either the transition happens
before a token arrives to the place or after the token left that place.
This is possible by adding direct interleaving dependencies
($\intrel$) between edges.
Namely, if $p$ inhibits $t$ then either $t$ happens before the
transition putting token in $p$ or after the transition taking the
token from $p$.

\subsection{Channelled Transition Systems (CTS)}
\label{subsec:cts}
A \emph{Channelled Transition System} (CTS) is a tuple of the form
$\mathcal{T}=\langle C, \Lambda,B,S, \allowbreak
S_0,R,L,\listen\rangle$,
where $C$ is a set of channels, including the broadcast channel
($\toall$), $\Lambda$ is a \emph{state alphabet}, $B$ is a
\emph{transition alphabet}, $S$ is a set of states, $s_0\in S$
is an initial state, $R\subseteq S\times B \times S$ is
a transition relation, $L:S\rightarrow \Lambda$ is a labelling
function, and $\listen:S \rightarrow 2^C$ is a channel-listening
function such that
for every $s\in S$ we have $\toall \in \listen(s)$.
That is, a CTS is listening to the broadcast channel in every state.
We assume that $B = B^+ \times \{!,?\} \times C$, for
some set $B^+$.
That is, every transition labeled with some
$b\in B$ is either a message send ($!$)
or a message receive ($?$) on some channel $c\in C$.

Given $(b^+,!,c)\in B$ we write $?(b^+,!,c)$ for $(b^+,?,c)$ and
$ch(b^+,{-},c)$ for $c$.
That is, $?(b)$ is the
corresponding receive transition of a send transition $b$ and $ch(b)$
is the channel of $b$.

For a receive transition $b=(b^+,?,c)$ and a state $s\in S$ we write
$s\rightarrow_b$ if $c\in
\listen(s)$ and there is some $s'$ such that $(s,b,s')\in R$.
That is, $s$ is listening on channel $c$ and can participate, i.e., has
an outgoing receive transition for $b$.
We write $s\not \rightarrow_b$ if $c\in\listen(s)$ and it is
not the case that $s\rightarrow_b$. That is, $s$ is listening on
channel $c$
and is not able to participate.

A \emph{history} $h=s_0,\ldots,s_n$ is a
finite sequence of states  such that $s_0\in S_0$
and for every $0\leq i <n$ we have that $(s_i,b_i,s_{i+1})\in R$ for
some $ b_i\in B$.
The length of $h$ is $n+1$, denoted $|h|$.
For convenience we generalise notations applying to states to apply to
histories.
For example, we write $c\in \listen(h)$ when $c\in
\listen(s_n)$, $h\rightarrow_b$ when $s_n\rightarrow_b$ and
$h\not\rightarrow_b$ for $s_n\not\rightarrow_b$.
Similarly, if $h=s_0,\ldots, s_n$ and
$h'= s_0,\ldots, s_n,s_{n+1}$ where $(s_n,b_n,s_{i+1})\in R$, we write
$(h,b_n,h') \in R$.
Let $\hist(\mathcal{T})$ be the set of all histories of $\mathcal{T}$.
An \emph{execution} $\pi=s_0,b_0,s_1\ldots$ is an infinite
sequence such that for every $i\geq 0$ we have
$(s_i,b_i,s_{i+1})\in R$ and $b_i\in B$.
Thus, every prefix of $\pi$ (projected on states) is a history.


The linear semantics for CTS is given by a parallel composition
operator over a set of CTSs.
We include the full definition in appendix and refer the
reader to \cite{alrahman2021modelling}.
Intuitively, multicast channels are blocking. All agents who are
listening to the channel must be able to participate in the
communication in order for a send to be possible.
The broadcast channel, on the other hand, is non-blocking.
Agents always listen to the broadcast channel.
However, if they cannot participate in a communication it still goes on
without them.

\begin{toappendix}
	\subsection{Composition for Channeled Transition Systems}
	We include the definition of the parallel composition operator over
	CTS.
	A parallel composition of CTSs is again a CTS.

	\begin{definition}[Parallel Composition]\label{def:par}
		Given two CTS
		$\mathcal{T}_i=\langle C_i, \Lambda_i,B_i,S_i, \allowbreak
		s^i_0,R_i,L_i,\listen^i\rangle$, where $i\in \{1,2\}$ their
		composition
		$\mathcal{T}_1\parallel \mathcal{T}_2$ is the following CTS
		$\mathcal{T}=\langle C, \Lambda,B,S, \allowbreak
		s_0,R,L,\listen\rangle$,
		where the components of $\mathcal{T}$ are:

		\noindent
		\begin{minipage}{0.5\textwidth}
			\begin{itemize}
				\item[]
				\item
				$C = C_1 \cup C_2$
				\item
				$B = B_1 \cup B_2$
				\item
				$s_0 = (s_0^1,s_0^2)$
				\item
				$\listen(s_1,s_2) = \listen^1(s_1)\cup \listen^2(s_2)$
			\end{itemize}
		\end{minipage}
		\begin{minipage}{0.5\textwidth}
			\begin{itemize}
				\item
				$\Lambda = \Lambda_1 \times \Lambda_2$
				\item
				$S = S_1\times S_2$
				\item
				$L(s_1,s_2) = (L_1(s_1),L_2(s_2))$
			\end{itemize}
		\end{minipage}\hfill

		\begin{itemize}
			\item
			$R = $
			{\small
				\[
				\begin{array}{r r}
					\left \{((s_1,s_2),(\upsilon,!,c),(s'_1,s'_2))
					\left |~
					\begin{array}{l r }
						(s_1,(\upsilon,!,c),s'_1) \in R_1,
						c\in \listen^2(s_2)  \mbox{ and }
						(s_2,(\upsilon,?,c),s'_2) \in R_2 & \mbox{ or}
						\\[2ex] 
						(s_1,(\upsilon,?,c),s'_1) \in R_1,c\in
						\listen^1(s_1), \mbox{ and }
						(s_2,(\upsilon,!,c),s'_2) \in R_2 & \mbox{ or}
						\\[2ex] 
						(s_1,(\upsilon,!,c),s'_1) \in R_1, c\notin
						\listen^2(s_2), \mbox{
							and }
						s_2=s'_2 & \mbox{ or} \\[2ex] 
						c\notin \listen^1(s_1), s_1=s'_1, \mbox{ and }
						(s_2,(\upsilon,!,c),s'_2) \in R_2 
					\end{array}
					\right . \right \} & \cup
					\\
					\left \{((s_1,s_2),(\upsilon,?,c),(s'_1,s'_2))
					\left |~
					\begin{array}{l r }
						c\in \listen^1(s_1), (s_1,(\upsilon,?,c),s'_1)
						\in R_1, c\in
						\listen^2(s_2)\\ \mbox{ and }
						(s_2,(\upsilon,?,c),s'_2) \in R_2 & \mbox{ or}
						\\[2ex] 
						(s_1,(\upsilon,?,c),s'_1) \in R_1, c\notin
						\listen^2(s_2), \mbox{
							and }
						s_2=s'_2 & \mbox{ or} \\[2ex] 
						c\notin \listen^1(s_1),  s_1=s'_1, \mbox{ and }
						(s_2,(\upsilon,?,c),s'_2) \in R_2 
					\end{array}
					\right . \right \}  & \cup\\
				\end{array}
				\]}
			\[
			\begin{array}{l r}\small
				\left
				\{((s_1,s_2),(\upsilon,\gamma,\toall),(s'_1,s'_2)) \left
				 |~
				\begin{array}{l r }
					\begin{array}{l r }
						\gamma\in\{ !,?\},
						(s_1,(\upsilon,\gamma,\toall),s'_1)\in R_1,
						s_2=s_2' \mbox{ and }\\
						\forall s_2'' ~.~
						(s_2,(\upsilon,?,\toall),s''_2) \notin R_2 &
						\hfill \mbox{ or} \\[2ex]
						\multicolumn{2}{l}{\gamma\in\{ !,?\}, s_1=s_1',
							\forall s_1'' ~.~
							(s_1,(\upsilon,?,\toall),s''_1) \notin R_1\
							\mbox{and}}\\
						(s_2,(\upsilon,\gamma,\toall),s'_2)\in R_2,
					\end{array}  
				\end{array}
				\right . \right \}
				\\[2ex]
			\end{array}
			\]

		\end{itemize}
	\end{definition}

	The transition relation $R$ of the composition defines two modes of
	interactions, namely multicast and broadcast.
	In both interaction modes, the composition $\mathcal{T}$ sends a
	message $(\upsilon,!,c)$ on channel $c$ (i.e.,
	$((s_1,s_2),(\upsilon,!,c),(s'_1,s'_2))\in R$) if either
	$\mathcal{T}_1$ or $\mathcal{T}_2$ is able to generate this message,
	i.e, $(s_1,(\upsilon,!,c),s'_1)\in R_1$ or
	$(s_2,(\upsilon,!,c),s'_2)\in R_2$.

	Consider the case of a multicast channel.
	A multicast is blocking. Thus, a multicast message is sent if either
	it is received or the channel it is sent on is not listened to.
	Suppose that a message originates from $\mathcal{T}_1$, i.e.,
	$(s_1,(\upsilon,!,c),s'_1)\in R_1$.
	Then, $\mathcal{T}_2$ must be able to
	either receive the message or, in the case that $\mathcal{T}_2$ does
	not listen to the channel, discard it.
	CTS $\mathcal{T}_2$ receives if $(s_2,(\upsilon,?,c),s'_2)\in R_2$.
	It discards if $c\notin\listen^2(s_2)$ and $s_2=s'_2$.
	The case of $\mathcal{T}_2$ sending is dual.
	Note that $\mathcal{T}_2$ might be a composition of other
	CTS(s), say $\mathcal{T}_2=\mathcal{T}_3\|\mathcal{T}_4$. In this
	case, $\mathcal{T}_2$ listens to channel $c$ if at least one of
	$\mathcal{T}_3$ or $\mathcal{T}_4$ is listening.
	That is, it could be that either
	$c\in(\listen(s_3)\cap\listen(s_4))$,
	$c\in(\listen(s_2)\backslash\listen(s_3))$, or
	$c\in(\listen(s_2)\backslash\listen(s_4))$.
	In the first case, both must receive the message.
	In the latter cases, the listener receives and the non-listener
	discards.
	Accordingly,
	when a message is sent by one system, it is propagated to all other
	connected systems in a joint transition.
	A multicast is indeed blocking because a connected system cannot
	discard an incoming message on a channel it is listening to.
	More precisely, a joint transition
	$((s_1,s_2),(\upsilon,!,c),(s'_1,s'_2))$ where $c\in\listen(s_2)$
	requires that $(s_2,(\upsilon,?,c),s'_2)$ is
	supplied. In other words, message sending is blocked until all
	connected receivers are ready to participate in the interaction.
	Clearly, the latter correspond to inhibition arcs in Petri nets.

	Consider now a broadcast.
	A broadcast is non-blocking.
	Thus, a broadcast message is either received or discarded.
	Suppose that a message originates from $\mathcal{T}_1$, i.e.,
	$(s_1,(\upsilon,!,\toall),s'_1)\in R_1$.
	If $\mathcal{T}_2$
	is receiving, i.e., $(s_2,(\upsilon,?,\toall),s'_2)\in R_2$ the
	message is sent.
	However,
	by definition, we have that $\toall\in\listen(s)$ for every $s$ in a
	CTS.
	Namely, a system may not disconnect the broadcast channel
	$\toall$.
	For this reason, the last part of the transition relation $R$
	defines a special case for handling (non-blocking)
	broadcast.
	Accordingly, a joint transition
	$((s_1,s_2),(\upsilon,\gamma,\toall),(s'_1,s'_2))\in R$ where
	$\gamma\in\set{!,?}$ is always possible and may not be blocked by
	any
	receiver.
	In fact, if ($\gamma=\ !$) and $(s_1,(\upsilon,!,\toall),s'_1)\in
	R_1$
	then the joint transition is possible whether
	$(s_2,(\upsilon,?,\toall),s'_2)\in R_2$ or not.
	In other words, a broadcast can happen even if there are no
	receivers.
	Furthermore, if ($\gamma=\ ?$) and
	$(s_1,(\upsilon,?,\toall),s'_1)\in R_1$ then also the joint
	transition
	is possible regardless of the other participants. In other words, a
	broadcast is received only by interested participants.
\end{toappendix}

\begin{figure}[bt]
	\centering
	\begin{tabular}{c@{\qquad\qquad\qquad}c@{\qquad\qquad}c}
		\scriptsize
		\begin{tikzpicture}[node distance=1.3cm]
			\tikzstyle{place}=[circle,thick,draw=blue,fill=white!20,minimum
			size=4mm]

			\node[place,label=right:$\set{\toall}$]
			(p7)  {1};

			\node[place,label=right:$\set{\toall}$]
			(p8) [below of=p7]  {2};

			\draw[-latex,thick] (p7) -- node[left] {$(v_1,!,c)$}(p8);

		\end{tikzpicture} &\quad
		\scriptsize
		\begin{tikzpicture}[node distance=1.3cm]
			\tikzstyle{place}=[circle,thick,draw=blue,fill=white!20,minimum
			size=4mm]

			\node[place,label=right:$\set{\toall}$]
			(p1)  {1};

			\node[place,label=right:$\set{\toall,c}$]
			(p3) [below of=p1] {2};

			\node[place,label=right:$\set{\toall}$]
			(p5) [below of=p3] {3};

			\draw[-latex,thick] (p1) -- node[left] {$(v_2,!,d)$} (p3);
			\draw[-latex,thick] (p3) -- node[left] {$(v_3,!,e)$} (p5);

		\end{tikzpicture} &\hspace{-5mm}
		\scriptsize
		\begin{tikzpicture}[node distance=1.3cm]
			\tikzstyle{place}=[circle,thick,draw=blue,fill=white!20,minimum
			size=4mm]

			\node[place,label=right:$\set{\toall,d}$]
			(p2)  {1};

			\node[place,label=right:$\set{\toall,e}$]
			(p4) [below of=p2] {2};

			\node[place,label=right:$\set{\toall}$]
			(p8) [below of=p4]  {3};
			\node[place,label=right:$\set{\toall}$]
			(p9) [right of=p8]  {4};
			\draw[-latex,thick] (p2) -- node[left] {$(v_2,?,d)$} (p4);
			\draw[-latex,thick] (p4) -- node[left] {$(v_3,?,e)$} (p5);
			\draw[-latex,thick] (p4) -- node[right] {$(v_4,!,b)$} (p9);

		\end{tikzpicture}

	\end{tabular}\\[2ex]
	\begin{tabular}{c@{\qquad\qquad\qquad\qquad}c@{\qquad\qquad\qquad}c}
		(a) Agent $\mathcal{T}_1$ &(b) Agent $\mathcal{T}_2$& (c) Agent
		$\mathcal{T}_3$

	\end{tabular}
	\caption{CTS representation of the running example.}
	\label{fig:cts}
\end{figure}
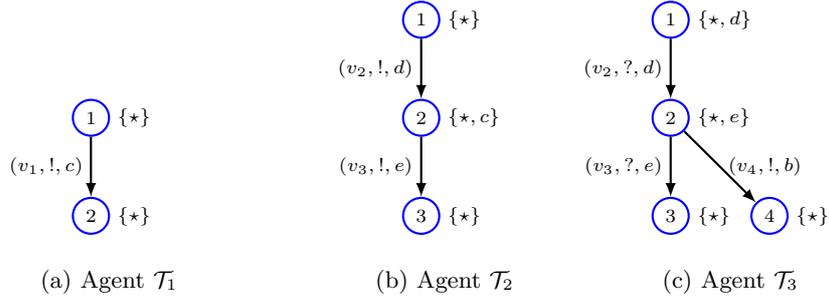

The PTI-net in Fig.~\ref{fig:petriex}(a) can be modelled as the
parallel composition of the CTSs in Fig.~\ref{fig:cts}, where we label
states with the listening function.
Starting from the initial states, we have that either $(v_1,!,c)$ or
$(v_2,!,d)$ can be sent. The former is an individual transition of
agent $\mathcal{T}_1$ while the latter is a joint transition between
$\mathcal{T}_2$ and $\mathcal{T}_3$ where $\mathcal{T}_2$ sends and
$\mathcal{T}_3$ receives. Note that $\mathcal{T}_3$ is initially
connected to channel $d$. If $(v_2,!,d)$ is scheduled first then the
listening function of both $\mathcal{T}_2$ and $\mathcal{T}_3$ is
reconfigured where $\mathcal{T}_2$ starts listening to channel $c$ and
$\mathcal{T}_3$ starts listening to $e$. This way, $(v_1,!,c)$ is
blocked until $(v_3,!,e)$ is sent.
It is not hard to see that a reconfiguration due to changes in the
listening function is equivalent to token passing. However, here we can
model a more interesting compositional interactions with meaningful
message exchange.

Now, everything is in place to define the labelled partial order
semantics of a CTS.
Consider a system $\mathcal{S}=\mathcal{T}_1 \parallel \cdots \parallel
\mathcal{T}_n$,
where $\mathcal{T}_i=\langle
C_i,\Lambda_i,B_i,S_i, S_0^i, R_i,L_i,\listen_i\rangle$.
We denote $C=\bigcup_i C_i$, and
$B=\bigcup_i B_i$.

\begin{definition}[\lpo-computation]
	\label{def:lpo-computation cts}
	A \emph{computation} of $\mathcal{S}$ is an $\mathsf{LPO}$
	$(O,\comrel,\intrel,\Sigma,\Upsilon,L)$, where
	$V\subseteq \bigcup_i \hist(\mathcal{T}_i)$, $\Sigma=V$,
	$\comrel = \rightarrow_s \biguplus \rightarrow_r$ is the
	disjoint union of the send and receive relations,
	$\Upsilon=\{(\upsilon,!,c)\in B\}$, and for $h\in V$ we have
	$L(h)=h$.
	In addition we require the following:
	\begin{enumerate}
		\item[C1.]
		The edge $e_\epsilon$ such that $L(e_\epsilon)=(b,!,\toall)$ is the
		unique minimal element according to $\leq$.
		For every $i$, we have $s^0_i \in V$ and $e_\epsilon\rightarrow_r
		s^0_i$.
		\item[C2.]
		If $h\in V\cap \hist(\mathcal{T}_i)$ there is
		a unique $e\in E$ such that $e\comrel h$.
		If $|h|>1$, there is also a unique $h'\in V$ such that $h'\comrel e$
		and either  $(h',L(e),h)\in R_i$ or $(h',?(L(e)),h)\in R_i$.
		\item[C3.]
		For every $h \in V$ there is at most one $e\in E$ such that
		$h \comrel e$.
		\item[C4.]
		For every $e\in E\setminus\{e_\epsilon\}$ there is $I\subseteq [n]$
		such that all the following hold:
		\begin{enumerate}[(a)]
			\item
			For every $i\in I$ we have $|\pre{e}\cap \hist(\mathcal{T}_i)|=1$
			and $|\post{e}\cap \hist(\mathcal{T}_i)|=1$.
			\item
			There is a unique $i\in I$ and $h,h'\in V\cap
			\hist(\mathcal{T}_i)$
			such that $(h,L(e),h')\in R_i$ and $h\rightarrow_s e
			\rightarrow_s h'$ and for every	$i' \in I \setminus\{i\}$ there
			are $h'',h'''\in V \cap \hist(\mathcal{T}_{i'})$ such that
			$h''\rightarrow_r e\rightarrow_r h'''$ and $(h'',?(L(e),h''')\in
			R_{i'}$.
			\item
			If $L(e)=(\upsilon,!,c)$ for $c\neq \toall$ then for every
			$h\in V$ such that
\iflisten
			$c \in \listen(h)$
\else
			$h \not\rightarrow_{?(L(e))}$
\fi
			we have $h\leq e$ or $e\leq h$.
			\item
			If $L(e)=(\upsilon,!,\toall)$ then for every $h\in V$ such
			that $h\rightarrow_{?(L(e))}$ we have $h\leq e$ or $e\leq
			h$.
		\end{enumerate}
		\item[C5.]
		For every $e\neq e'$ such that $ch(e)=ch(e')$
		we
		have $e\leq e'$ or $e'\leq e$.
		\item[C6.]
		If $e\intrel e'$ then there is some
		$h=s_0,\ldots,
		s_j$
		such that one of the following holds:
		\begin{enumerate}
			\item
			$ch(e)=ch(e')$.
		  \np{Do ii.-v.  need to be $\comrel$ or
			$\rightarrow_r$? }
			\item
			$L(e')=(\upsilon,!,c)$ for $c\neq \toall$,
			$h\comrel e$ and
\iflisten
			$ch(L(e')) \in \listen(h)$.
\else
			$h\not\rightarrow_{?(L(e'))}$.
\fi
			\item
			$L(e)=(\upsilon,!,c)$ for $c\neq \toall$,
			$e' \comrel h$ and
\iflisten
			$ch(L(e')) \in \listen(h)$.
\else
			$h\not\rightarrow_{?(L(e))}$.
\fi
			\item
			$L(e')=(\upsilon,!,\toall)$, $h \comrel e$
			and
			$h\rightarrow_{?(L(e'))}$.
			\item
			$L(e)=(\upsilon,!,\toall)$, $e'\comrel h$
			and
			$h\rightarrow_{?(L(e))}$.
		\end{enumerate}
	\end{enumerate}
\end{definition}

Note that an $\mathsf{LPO}$ computation relates histories
of individual CTSs, and thus allows to draw
relations among finite sequences of individual computation
steps of one CTS (or a group of CTSs) with respect to others;
Furthermore, a CTS is always listening to the broadcast channel,
and thus, it becomes mandatory to order broadcast messages that
enable/disable participation to each other.

More precisely, C1 ensures that a unique broadcast initiates all the
initial states of $\mathcal{T}_i$ for all $i$ and that nothing happens
before that.
As expected, C2 and C3 ensure that an \lpo
defines a unique resolution of a nondeterministic choice
in every single step.
Moreover, C4 models interactions, where (a) and (b) model
synchronisation while (c)-(f) model ordering due to schedule imposed by
using global resources and restrictions due to reconfiguration.
First, communications on the same channel must be ordered.
Then, a multicast must be ordered with respect to every
individual history that listens to it.
Furthermore,  a broadcast must be ordered with respect to every
individual history that can participate in it.
Clearly, the last two requirements are crucial to preserve the blocking
semantics of multicasts and the input enabledness of broadcasts.

Thus, for a multicast, if a history $h$ blocks the multicast execution
then either $h$ can be extended so that the multicast is released or
the multicast happens directly before $h$ is reached.
We solve this by adding a strict ordering between multicasts.
The same holds for a broadcast, but in this case we handle input
enabledness of broadcast rather.

We will use $\comp(\mathcal{S})$ for $\mathcal{S}$ being a
Petri net or CTS,
to denote the set of \lpo computations of $\mathcal{S}$.
\section{Partial Order with Glue}\label{sec:glue}

In this section we extend labeled partial orders with
\emph{glue}.
Intuitively, two elements are glued from the point of view of another
element if they both happen either before or after said element.

\begin{definition}[Glue relation]\label{def:glue}
  A \emph{Glue} over a set $O$ and a relation $\rightarrow_c\subseteq
  O\times O$ is a relation $R\subseteq \rightarrow_c$.
\end{definition}

Intuitively, a glue relation $R$ over the set $O$ and a relation
$\rightarrow_c$ defines pairs of elements that are glued
together.

\begin{definition}[Glued \lpo]\label{def:gluedpo}
A \emph{glued labeled partial order} (\glpo, for short) is
$\rulename{lpg}=(P,\mathcal{G},\mathcal{E})$, where $P=(O=V\biguplus
E,\rightarrow_c,\rightarrow_i,\Sigma,\Upsilon, L)$ is an \lpo,
$\mathcal{G}=\set{G_1,\ldots, G_k}$ is a set of Glue
relations over $O$ and $\rightarrow_c$, and
$\mathcal{E}:\Upsilon\hookrightarrow \mathcal{G}$ labels elements in $E$
(through their edge labels) by glue relations.

\end{definition}

\begin{definition}[\rulename{\glpo}-refinement]\label{def:refine}
An \lpo
$\rulename{lpo}=(O,\rightarrow_c,\rightarrow_i,\Sigma,\Upsilon,L)$
where $O=V\biguplus E$ refines a \glpo
$\rulename{lpg}=(P_g,\mathcal{G},\mathcal{E})$, denoted $\rulename{lpo}\preceq \rulename{lpg}$, where
$P_g=(O,\rightarrow_c,\rightarrow_i^g\Sigma,\Upsilon,L)$ if the
following conditions hold:
\begin{itemize}
\item
  For every $e\in E$ and  $(a,b) \in \mathcal{E}(L(e))$ we have
	$e\leq a$ or $b \leq e$.
  \item
   ${\rightarrow^g_i}\subseteq
  {\rightarrow_i}$ and $(e,e')\in ({\rightarrow_i\setminus
  \rightarrow^g_i})$ implies $(e',v)\in \mathcal{E}(L(e))$ for some
  $v$ or $(v,e)\in \mathcal{E}(L(e'))$ for some
  $v$.
\end{itemize}
\end{definition}

That is, the two share the relation $\rightarrow_c$, the relation
$\rightarrow_i^g$ is
preserved and extended by extra interleaving to capture the glue.
In order to respect the glue, an edge that is glued to a pair $(a,b)$
must happen either before $a$ or after $b$.

We show now that \glpos enable to remove parts of
the interleaving
order relation for both PTI-nets and CTSs.
\glpos capture better reconfiguration by
combining multiple order
choices due to the same reconfiguration into the same g-computation.

\newcommand{\removethissection}{
\subsection{Experimental Non-Emptiness of \glpo}
\np{Very drafty and unsure.}

We want to say that the following is not circular:
If $e$ is

If $ch(e)=c$ then every $(v,e')$ such that $c\in listen(v)$ is in $G(e)$
Consider a glue relation $G_i$.
We want to make sure that the following does not happen:
$(v_1,e_1), \ldots, (v_n,e_n)$ such that $\epsilon(L(e))$

When does a \glpo have \lpos refining it?
A glue relation is only important in a \glpo if there is some edge
that is labeled by that glue.
Consider an important glue relation and think about the induced
equivalence relation between glued elements.
Then, no minimal element can be equivalent to a maximal element and no
minimal element can have an infinite equivalence class.

Also, consider a more elaborate equivalence relation.
Consider a pair $(v,e)$ such that $(v,e) \in R$ for some glue relation
$R$.
Consider an edge $e'$ such that $R$ is the label of $e'$ and let
$(v',e')$ be in some relation $R'$ as well.
And so on.
If the transitive closure of that includes a cycle then this is bad.

Consider a \glpo where these bad things to not happen.
We have to show that we can find a \lpo that refines it.
We add edge-edge connections to manage the glue ordering.
This should have two elements:
choose ``blocked'' edges order them
according to the relation in (b) and as there is no cycle there has to
be a minimal element. Schedule that one before all the others.

If the scheduling does not cover the entire \glpo we should be able to
find an infinite chain as in (a).
}

\subsection{Glue Computations for PTI-nets}

Let $N = \conf{P, T, F, I}$
be a PTI-net and $m_0$ its initial marking.
We now define a \emph{g-computation}. The differences
from the definition of \lpo
(Definition~\ref{def:lpo-computation})
are highlighted with a $\qt{*}$.
\begin{definition}[\rulename{g}-computation]
	\label{def:g-computation}
	A \emph{g-computation} of $N$ is a \glpo
	$(P,\mathcal{G},\mathcal{E})$, where
	$P=(O,\comrel,\Sigma,\Upsilon,L)$, the
	components $V$, $E$, $\Sigma$, $\Upsilon$, and $L$
	are as for \lpo,  and the following holds.
	\begin{enumerate}
		\item[{\color{white}$^*$}N1.]
		The t-history $(\emptyset,t_\epsilon)$ is the unique minimal
		element according to $\leq$.
		\item[{\color{white}$^*$}N2.]
		For a p-history $v=(e,p,i)\in V$ we have $e\in E$ and $e$ is
		the unique edge such that $e\comrel v$.
		\item[{\color{white}$^*$}N3.]
		For a p-history $v=(h,p,i)\in V$, let $e_1,\ldots, e_j$ be the
		t-histories such that $v\comrel e_j$.
		Then, for every $j$ we have $e_j(v)>0$ and $\sum_j e_j(v)\leq
		i$.
		That is, $v$ leads to t-histories that contain it with
		the multiplicity of $v$ being respected.
		\item[$^*$N4.]
		For every $e\in E$, where $e=(\{(v_1,i_1),\ldots, (v_n,i_n)\},t)$
		the following holds:
		\begin{enumerate}
			\item[$^*$(a)]
			$\pre{e} =\{v_1,\ldots, v_n\}$ and
			$\post{e} =\{(e,p,\post{t}(p)) ~|~
			\post{t}(p)>0\}$.
		\end{enumerate}
		\item[$^*$N5.]
		For every $t\in T$  we have:
		$$\begin{array}{l@{}l}
			\mathcal{E}(t) ~=~ &
			\{ (v,e) ~|~ v\comrel e \mbox{ and } (L(v),t)\in I\}\ \cup\\
			& \{(e,v) ~|~ e\comrel v \mbox{ and } (L(v),t)\in I\}
		\end{array} $$
	\end{enumerate}
\end{definition}
That is, we drop $\intrel$ and assign each
inhibited event (or transition) with a glue relation.
Namely, for every transition $t$ add all \emph{existing} ingoing and
outgoing transitions of places that inhibit $t$.

We use $\comp_g(N)$ to denote the set of g-computations
of Petri net $N$.

\begin{toappendix}
\begin{lemma}\label{lem:1stPN}
	Given an \lpo $\pi\in\comp(N)$, there exists a
	corresponding  \glpo $\lfloor \pi\rfloor\in\comp_g(N)$
	such that $\pi \preceq \lfloor \pi\rfloor$.
\end{lemma}

\begin{proof}
	Let $\lfloor \pi \rfloor$ be the \glpo obtained from $\pi$ by using
	$\comrel$ of $\pi$, setting
	$\intrel^g=\emptyset$, and adding the
	glue relations according to Definition~\ref{def:g-computation}.

	We have to show that the conditions of
	Definition~\ref{def:refine} hold.
	Note that by construction both $\pi$ and $\lfloor \pi \rfloor$ agree
	on $\comrel\subseteq V\times E \cup E\times V$ and only disagree in
	terms of $\intrel \subseteq E\times E$ and the glue.

	Consider some $t\in T$ and $(a,b)\in \mathcal{E}(t)$.
	We have to show that $e \leq a$ or $b\leq e$.
	By definition we know that $a\comrel b$.
	We have the following cases.

	\begin{itemize}
		\item
		If $a\in V$ and $b\in E$ then $(L(a),t)\in I$.
		By $N4(b)$ in definition~\ref{def:lpo-computation} we
		have that either $e \leq a$ or $a\leq e$.
		If $e\leq a$ we are done.
		If $a\leq e$ then from $a\comrel b$ it follows that either
		$e=b$ or $b<e$.
		\item
		If $a\in E$ and $b\in V$ then by definition $(L(a),t)\in I$.
		By $N4(b)$ in definition~\ref{def:lpo-computation} we
		have that either $e \leq b$ or $b\leq e$.
		If $b\leq e$ we are done.
		If $e\leq b$ then from $a\comrel b$ it follows that either
		$e=a$ or $e<a$.
	\end{itemize}

	Consider some $(e,e')\in \intrel$.
	We have to show that either $(e',v)\in \mathcal{E}(L(e))$ for some
	$v$ or $(v,e)\in \mathcal{E}(L(e'))$.
	By definition, we have that there exists $v\in V$ such that one of
	the following holds.
	\begin{itemize}
	\item $(L(v),L(e'))\in I$ and $v\comrel e$. By $^*N5$
	in Definition~\ref{def:g-computation}, we have that
	$(v,e)\in\mathcal{E}(L(e'))$ as required.

	\item $(L(v),L(e))\in I$ and $e'\comrel v$. By $^*N5$
	in Definition~\ref{def:g-computation}, we have that
	$(e',v)\in\mathcal{E}(L(e))$ as required.
	\end{itemize}
\end{proof}

\begin{lemma}\label{lem:2ndPN}
	Given a \glpo $ \pi_1\in\comp_g(N)$ and  an \lpo $
	\pi_2$ such that $ \pi_2 \preceq \pi_1$ then
	$ \pi_2\in\comp(N)$.
\end{lemma}

\begin{proof}
	Given that $\pi_2 \preceq \pi_1$, it follows that both
	$\pi_1$ and $\pi_2$ agree on
	$\comrel\subseteq V\times E \cup E\times V$ and only disagree in
	terms of $\intrel \subseteq E\times E$ and the glue.

	It is sufficient to prove that $N4$, items (b) and (c) in Definition~\ref{def:lpo-computation}
	hold for $\pi_2$. Consider some $e\in E$.
	We have the following cases.

	\begin{itemize}
		\item
		Consider some $v\in V$ and $e\in E$ such that
		$L(v)\in\inhpre{L(e)}$.
		In order to show that $\pi_2\in\comp(N)$ we have to show that
		$e\leq v$ or $v\leq e$.
		Let $e'$ be the edge such that $e'\comrel v$.
		By Definition~\ref{def:g-computation} ($^*N5$) we have that
		$(e',v) \in \mathcal{E}(L(e))$.
		By refinement, we have that either $e \leq e'$, which implies $e\leq v$,
		or $v\leq e$ as required.
		\item
		Consider some $e'\in E$ such that $e\intrel e'$.
		By refinement, we have one of the following cases holds.
		\begin{itemize}
			\item
			$(e',v)\in \mathcal{E}(L(e))$ for some $v$. By
			Definition~\ref{def:g-computation} ($^*N5$), we have
			that $e'\comrel v$ and $(L(v),L(e))\in I$ as required.
			\item
			$(v,e)\in \mathcal{E}(L(e'))$ for some $v$. By
			Definition~\ref{def:g-computation} ($^*N5$), we have
			that $v\comrel e$ and $(L(v),L(e'))\in I$ as required.
		\end{itemize}
	\end{itemize}
\end{proof}
\end{toappendix}

\begin{theoremrep}
	\label{thm:language glpo pn}
	Given a PTI-net $N$,
	$\comp(N) = \set{\pi \mid \pi \preceq \pi_g \wedge
	\pi_g\in\comp_g(N)}$.
\end{theoremrep}

\begin{proof}
	The proof follows directly from Lemma~\ref{lem:1stPN} and
	Lemma~\ref{lem:2ndPN}.
\end{proof}

\subsection{Glue Computations for CTSs}
Consider a system $\mathcal{S}=\mathcal{T}_1 \parallel \cdots \parallel
\mathcal{T}_n$,
where $\mathcal{T}_i=\langle
C_i,\Lambda_i,B_i,S_i, S_0^i, R_i,L_i,\listen_i\rangle$.
We denote $C=\bigcup_i C_i$ and
$B=\bigcup_i B_i$.

We now define a \emph{g-computation} for CTS.
As before, the differences from the definition of \lpo
(Definition~\ref{def:lpo-computation cts}) are highlighted
with a $\qt{*}$.

\begin{definition}[\rulename{g}-computation]
	\label{def:g-computation cts}
	A \emph{g-computation} of $S$ is a \glpo
	$(P,\mathcal{G},\mathcal{E})$, where
	$P=(O,\intrel,\comrel,\Sigma,\Upsilon,L_V,L_E)$
    and $V$, $E$, $\Sigma$, $\Upsilon$, and $L$ are as
    before, $\rightarrow_c=\rightarrow_s \biguplus
    \rightarrow_r$, and in addition:
   	\begin{enumerate}
		\item[{\color{white}$^*$}C1.]
		The edge $e_\epsilon$ such that $L(e_\epsilon)=(b,!,\toall)$ is the
		unique minimal element according to $\leq$.
		For every $i$, we have $s^0_i \in V$ and $e_\epsilon\rightarrow_r
		s^0_i$.
		\item[{\color{white}$^*$}C2.]
		If $h\in V\cap \hist(\mathcal{T}_i)$ there is
		a unique $e\in E$ such that $e\comrel h$.
		If $|h|>1$, there is also a unique $h'\in V$ such that $h'\comrel e$
		and either  $(h',L(e),h)\in R_i$ or $(h',?(L(e)),h)\in R_i$.
		\item[{\color{white}$^*$}C3.]
		For every $h \in V$ there is at most one $e\in E$ such that
		$h \comrel e$.
		\item[$^*$C4.]
		For every $e\in E\setminus\{e_\epsilon\}$ there is $I\subseteq [n]$
		such that all the following hold:
		\begin{enumerate}
			\item[(a)]
			For every $i\in I$ we have $|\pre{e}\cap \hist(\mathcal{T}_i)|=1$
			and $|\post{e}\cap \hist(\mathcal{T}_i)|=1$.
			\item[(b)]
			There is a unique $i\in I$ and $h,h'\in
			V\cap \hist(\mathcal{T}_i)$
			such that $(h,L(e),h')\in R_i$ and $h\rightarrow_s e
			\rightarrow_s h'$ and for every	$i' \in I \setminus\{i\}$ there
			are $h'',h'''\in V \cap \hist(\mathcal{T}_{i'})$ such that
			$h''\rightarrow_r e\rightarrow_r h'''$ and $(h'',?(L(e),h''')\in
			R_{i'}$.
		\end{enumerate}
		\item[{\color{white}$^*$}C5.]
		For every $e\neq e'$ such that
		$ch(e)=ch(e')$ we
		have $e\leq e'$ or $e'\leq e$.
		\item[$^*$C6.]
			If $e\intrel e'$ then the following holds:
			\begin{enumerate}
				\item[(a)]
				$ch(e)=ch(e')$.
			\end{enumerate}
		\item[$^*$C7.]
		For every $(\upsilon,!,c)\in B$  then
		\np{Should these be $\rightarrow_r$ or $\rightarrow_c$?
		}

		$$\begin{array}{l@{}l}
			\mathcal{E}((\upsilon,!,c)) ~=~ & \{ (h,e) ~|~ \mbox{for}\
			c\neq
			\toall,\ h\comrel e\ \mbox{and}\
\iflisten
			c \in \listen(h)\}\ \cup\\
\else
      h\not\rightarrow_{?(\upsilon,!,c)}\}\ \cup\\
\fi
			& \{(e,h) ~|~ \mbox{for}\  c\neq \toall, e \comrel h\
			\mbox{and}\
\iflisten
			c \in \listen(h)\}\ \cup\\
\else
			h\not\rightarrow_{?(\upsilon,!,c)}\}\ \cup\\
\fi
			& \{ (h,e) ~|~ \mbox{for}\ c=\toall,\ h\comrel e\
			\mbox{and}\
			h\rightarrow_{?(\upsilon,!,c)}\}\ \cup\\

			& \{(e,h) ~|~ \mbox{for}\  c= \toall, e \comrel h\
			\mbox{and}\
			h\rightarrow_{?(\upsilon,!,c)}\}\

		\end{array} $$
	\end{enumerate}
\end{definition}

We drop from the interleaving relation all order relations that
correspond to reconfiguration and keep only those that correspond to
the usage of a common resource.
Furthermore, we assign each  broadcast and multicast message with a
glue relation.
Namely, for every multicast $m$ add all \emph{existing}
ingoing and outgoing messages of histories that blocks $m$ execution;
for every broadcast $b$ add all \emph{existing} ingoing and outgoing
messages of histories that may participate in $m$. Note that, for the
case of broadcast, the rationale is that if such histories can
participate in a broadcast then they cannot be enabled independently
from the broadcast.
Notice that $^*C6$ adds one glue for every multicast
\emph{channel} but one for every broadcast
\emph{message}.

We use $\comp_g(\mathcal{S})$ to denote the set of g-computations of
CTS $\mathcal{S}$.

\begin{toappendix}
\begin{lemma}\label{lem:1st}
	Given an \lpo $\pi\in\comp(\mathcal{T})$, there exists a
	corresponding  \glpo $\lfloor \pi\rfloor\in\comp_g(\mathcal{T})$
	such that $\pi \preceq \lfloor \pi\rfloor$.
\end{lemma}

\begin{proof}
	Let $\lfloor \pi \rfloor$ be the \glpo obtained from $\pi$ by using
	the	partial order induced by $\comrel$ of $\pi$, by $\intrel$ of
	$\pi$ whenever $e\intrel e'$ implies $ch(e)=ch(e')$, and adding the
	glue relations according to Definition~\ref{def:g-computation cts}.

	Note that by construction both $\pi$ and $\lfloor \pi \rfloor$ agree
	on $\comrel\subseteq V\times E \cup E\times V$ and  only agree on
	$\intrel \subseteq E\times E$ whenever $e\intrel e'$
	implies $ch(e)=ch(e')$.
%
	We have to show that the conditions of
	Definition~\ref{def:refine} hold.

	Consider some $e\in E$ and $(a,b)\in \mathcal{E}(L(e))$.
	We have to show that $e\leq a$ or $b\leq e$.
	By definition we know that $a\comrel b$.
	We show that either $e \leq a$ or $b \leq e$.
	We have the following  cases.
	\begin{itemize}
		\item
		$ch(L(e))$ is a multicast channel:
		\begin{itemize}
			\item
			If $a\in V$ and $b\in E$ then by definition
\iflisten
			$ch(L(e)) \in \listen(a)$.
\else
			$a \not\rightarrow_{?(L(E))}$.
\fi
			By $C4(c)$ in definition~\ref{def:lpo-computation cts} we
			have that either $e \leq a$ or $a\leq e$.
			If $e\leq a$ we are done.
			If $a\leq e$ then from $a\comrel b$ it follows that either
			$e=b$ or $b<e$.
			\item
			If $a\in E$ and $b\in V$ then by definition
\iflisten
			$ch(L(e)) \in \listen(b)$.
\else
			$b\not\rightarrow_{?(L(e))}$.
\fi
			By $C4(c)$ in definition~\ref{def:lpo-computation cts} we
			have that either $e \leq b$ or $b\leq e$.
			If $b\leq e$ we are done.
			If $e\leq b$ then from $a\comrel b$ it follows that either
			$e=a$ or $e<a$.
		\end{itemize}
		\item
		$ch(L(e))$ is the broadcast channel:
		\begin{itemize}
			\item
			If $a\in V$ and $b\in E$ then $a\rightarrow_{?(L(e))}$.
			By $C4(d)$ in definition~\ref{def:lpo-computation cts} we
			have that either $e \leq a$ or $a\leq e$.
			If $e\leq a$ we are done.
			If $a\leq e$ then from $a\comrel b$ it follows that either
			$e=b$ or $b<e$.
			\item
			If $a\in E$ and $b\in V$ then $b\rightarrow_{?(L(e))}$.
			By $C4(d)$ in definition~\ref{def:g-computation cts} we
			have that either $e \leq b$ or $b \leq e$.
			If $b\leq e$ we are done.
			If $e\leq b$ then from $a\comrel b$ it follows that either
			$e=a$ or $e<a$.
		\end{itemize}
	\end{itemize}

	Consider $e,e'\in E$ such that $(e,e')\in\intrel^g$.
	By construction we have that  $(e,e')\in\intrel$, and thus
	$\intrel^g\subseteq\intrel$.
	Consider $e,e'\in E$ such that $(e,e')\in
	(\intrel\backslash\intrel^g)$.
	By $C6(b)-(e)$ in
	Definiton~\ref{def:lpo-computation cts} there exists $v\in V$
	such that one of the following holds.
  \begin{itemize}
    \item
    $ch(L(e'))\neq\toall$,
\iflisten
		$ch(L(e')) \in \listen(v)$ and
\else
    $v\not\rightarrow_{?(L(e'))}$ and
\fi
    $v\comrel e$.
    By $^*C7$ in Definition~\ref{def:g-computation cts}, we have that
    $(v,e)\in\mathcal{E}(L(e'))$ as required.

   	\item
   	$ch(e)\neq\toall$,
\iflisten
		$ch(L(e)) \in \listen(v)$ and
\else
   	$v\not\rightarrow_{?(L(e))}$ and
\fi
   	$e'\comrel v$. By $^*C7$
   	in Definition~\ref{def:g-computation cts}, we have that
   	$(e',v)\in\mathcal{E}(L(e))$ as required.

   	\item $ch(e')=\toall$, $v\rightarrow_{?(L(e'))}$ and
  	$v\comrel e$. By $^*C7$
  	in Definition~\ref{def:g-computation cts}, we have that
  	$(v,e)\in\mathcal{E}(L(e'))$ as required.

  	\item $ch(e)=\toall$, $v\rightarrow_{?(L(e))}$ and
  	$e'\comrel v$. By $^*C7$
  	in Definition~\ref{def:g-computation cts}, we have that
  	$(e',v)\in\mathcal{E}(L(e))$ as required.
	\end{itemize}
\end{proof}

\begin{lemma}\label{lem:2nd}
	Given a \glpo $ \pi_1\in\comp_g(\mathcal{T})$ and  an \lpo $
	\pi_2$ such that $ \pi_2 \preceq \pi_1$ then
	$ \pi_2\in\comp(\mathcal{T})$.
\end{lemma}

\begin{proof}
	Given that $\pi_2 \preceq \pi_1$, it follows that both
	$\pi_2$ and $\pi_2$ agree on $\comrel\subseteq
	V\times E \cup E\times V$ and  only agree on
	$\intrel \subseteq E\times E$ whenever $e\intrel e'$
	implies $ch(e)=ch(e')$.
	Hence, it is sufficient to prove that $C4(c)-(d)$ and $C6(b)-(e)$ in
	Definition~\ref{def:lpo-computation cts} hold
	for $\pi_2$.

	\noindent
	We prove $C4(c)-(d)$.
	Consider some $e\in E$.
	We have the following cases.

	\begin{itemize}
		\item
		$ch(L(e))$ is a multicast channel:

		Consider some $v\in V$  such that
\iflisten
		$ch(L(e))\in\listen(v)$.
\else
		$v \not\rightarrow_{?(L(e))}$.
\fi
		We have to show that $e\leq v$ or $v\leq e$.
		By Definition~\ref{def:g-computation cts} ($^*C7$), there is some
		$e'$ such that one of the following cases holds.
		\begin{itemize}
			\item
			$(v,e')\in\mathcal{E}(L(e))$ where $v\comrel e'$.
			By refinement, we have that if
			$(v,e')\in\mathcal{E}(L(e))$
			then either $e\leq v$ as required or $e'\leq e$, which implies
			that $v\leq e$.
			\item
			$(e',v)\in\mathcal{E}(L(e))$ where $e'\comrel v$.
			By refinement, we have that if
			$(e',v)\in\mathcal{E}(L(e))$
			then either $v\leq e$ as required or $e\leq e'$, which implies
			that $e\leq v$.
		\end{itemize}
		\item
		$ch(L(e))$ is a broadcast channel:

		Consider some $v\in V$  such that $v\rightarrow_{?(L(e))}$.
		We have to show that $e\leq v$ or $v\leq e$.
		By Definition~\ref{def:g-computation cts} ($^*C7$), there
		is some $e'$ such that one of the	following cases holds.
		\begin{itemize}
			\item
			$(v,e')\in\mathcal{E}(L(e))$ where $v\comrel e'$.
			By refinement, we have that if $(v,e')\in\mathcal{E}(L(e))$
			then either $e\leq v$ as required or $e' \leq e$, which implies
			that $e\leq v$.
			\item
			$(e',v)\in\mathcal{E}(L(e))$ where $e'\comrel v$.
			By refinement, we have that if $(e',v)\in\mathcal{E}(L(e))$
			then either $v \leq e$ as required or $e\leq e'$, which implies
			that $e\leq v$.
		\end{itemize}
	\end{itemize}
	We prove $C6(b)-(e)$.
	Consider $(e,e')\in \intrel$ such that $(e,e')\in
	(\intrel\backslash\intrel^g)$.
	By refinement, we have one of the following cases hold.
	\begin{itemize}
		\item
		$ch(L(e))$ is a multicast channel:
		\begin{itemize}
			\item
			$(e',v)\in \mathcal{E}(L(e))$ for some $v$.
			By Definition~\ref{def:g-computation cts} ($^*C7$), we have
			that $e'\comrel v$ and
\iflisten
			$ch(L(e)) \in \listen(v)$
\else
			$v\not\rightarrow_{?(L(e))}$
\fi
			as required.
			\item
			$(v,e)\in \mathcal{E}(L(e'))$ for some $v$.
			By Definition~\ref{def:g-computation cts} ($^*C7$), we have
			that $v\comrel e$ and
\iflisten
			$ch(L(e)) \in \listen(v)$
\else
			$v\not\rightarrow_{?(L(e'))}$
\fi
			as required.
		\end{itemize}
		\item
		$ch(L(e))$ is a broadcast channel:
		\begin{itemize}
			\item
			$(e',v)\in \mathcal{E}(L(e))$ for some $v$.
			By Definition~\ref{def:g-computation cts} ($^*C6$), we have
			that $e'\comrel v$ and $v\rightarrow_{?(L(e))}$ as required.
			\item
			$(v,e)\in \mathcal{E}(L(e'))$ for some $v$.
			By Definition~\ref{def:g-computation cts} ($^*C6$), we have
			that $v\comrel e$ and $v\rightarrow_{?(L(e'))}$ as required.
		\end{itemize}
	\end{itemize}
\end{proof}
\end{toappendix}

\begin{theoremrep}
	\label{thm:language glpo cts}
	Given a CTS $\mathcal{T}$,
	$\comp(\mathcal{T}) =  \set{\pi \mid \pi \preceq \pi_g \wedge
	\pi_g\in\comp_g(\mcal{T})}$.
\end{theoremrep}

\begin{proof}
	The proof follows by Lemma~\ref{lem:1st} and Lemma~\ref{lem:2nd}.
\end{proof}

\section{Separating Choice and  Reconfiguration-Forced
Interleaving}
\label{sec:separation}

We show that \glpos capture the differences between nondeterministic
choice, which corresponds to different \glpos, and interleaving
choices due to reconfiguration, which correspond to different ways to
refer to glue.
For both PTI-nets and CTS we show that distinct \glpos contain
different nondeterministic or order choices.

\subsection{Choice vs Interleaving in PTI-nets}

A choice is a situation where a set of tokens have exactly the same
history and they do a different exchange.

We show that every two distinct g-computations of the same net have a
set of tokens that ``see the difference''.
That is, they participate in a different transition in the two
g-computations.
This includes the option of tokens in one g-computation participating
in a transition and tokens in the other g-computation not continuing.

\begin{theoremrep}
  Given a Petri net $P$ and two different \glpos $G_1,G_2\in\comp_g(N)$
  then there exists a set of nodes $v_1,\ldots, v_n$ appearing in both
  $G_1$ and in $G_2$ such that one of the following holds:
  \begin{enumerate}
  	\item
  	There is a node $v_i$ such that the number of tokens not taken from
  	$v_i$ in $G_1$ and $G_2$ is different.
  	\item
  	\label{difference 2}
  	There is a set of p-histories $v_1,\ldots, v_n$ that participate in
  	some transition $t$ in $G_i$ but not in $G_{3-i}$.
  \end{enumerate}
  \label{thm:different glpos petri nets}
\end{theoremrep}

Notice that item~\ref{difference 2} includes the case where the
transition $t$ happens in both $G_1$ and $G_2$ but takes a different
number of tokens from every node.
This difference is indeed significant as the nodes communicate via the
identified transition and share the knowledge about the difference.

Theorem~\ref{thm:different glpos petri nets} is not true for \lpos.
This is already shown by the very simple examples in
Figure~\ref{fig:petriex}(b).
Indeed, in the two \lpos corresponding to each of the dashed arcs in the
figure all sets of nodes participate in exactly the same transitions.

\begin{appendixproof}
	We define the \emph{depth} of a history to be the maximal number of
	transitions taken by some token in the history.
	Formally, the depth $(\emptyset,t_\epsilon)$ is $0$.
	The depth of a p-history $(h,p,j)$ is $depth(h)+1$.
	For a t-history $e\in E$, let $\pre{h}$ be $\{e_1,\ldots, e_n\}$,
	then the depth of $e$ is $\max_j depth(e_j)$.
	Notice, that a t-history $e$ could have other edges in its preset.

	We order the elements in a \glpo by increasing depth. In addition,
	elements of the same depth are ordered so that edges appear before
	vertices and there is some arbitrary order between edges of the same
	depth and between vertices of the same depth.
	Clearly, in this order every element appears after all the elements
	that are smaller than it according to $\leq$.
	Indeed, if $a\comrel b$ or $a\intrel b$, then the depth of $b$ is at
	least the depth of $a$ plus one.
	As every element has a finite depth and there is a finite number of
	elements in every depth, it follows that this order is some
	linearisation of \emph{all} the elements in the \glpo.

	We prove the theorem by induction according to the order mentioned
	above.
	We are going to mark nodes and edges that appear in both $G_1$ and
	$G_2$.
	Nodes are marked by the number of tokens in them that we have not
	handled yet.
	When this number is $0$ the node is called closed.
	Otherwise, it is open.
	Edges are simply marked (or unmarked).
	For all marked nodes, we ``handle'' tokens that are participating
	in the same transitions in $G_1$ and $G_2$.
	Nodes could have tokens that do not participate in transitions.
	As we ``handle'' tokens we mark transitions continuing from the
	node as not forming part of the difference between $G_1$ and $G_2$.
	Once we mark nodes as closed they are also equivalent in $G_1$
	and $G_2$.
	As we go through the nodes in $G_1$ in induction order either we
	find a difference or, if not, the induction proves that $G_1$ and
	$G_2$ are equivalent in contradiction to the assumption.

	Both $G_1$ and $G_2$ have the t-history
	$h_\epsilon=(\emptyset,t_\epsilon)$ as minimal element. Mark it as
	closed.
	The p-histories of the form $(h_\epsilon,p,m_0(p))$ such that
	$m_0(p)>0$ are marked by $m_0(p)$.
	Clearly, as both $G_1$ and $G_2$ start from the initial marking
	$m_0$ both $G_1$ and $G_2$ have the same nodes marked and they
	have the same positive number of tokens.

	Assume that we have marked a prefix of $G_1$ and $G_2$ such that
	all closed nodes have all their outgoing transitions marked.
	Furthermore, the number marking a node is sufficient for all
	unmarked transitions existing from the node.
	Clearly, this is true of the marking of the minimal nodes.

	Suppose that there are some open nodes.
	Choose the minimal open node $v$ according to the induction order.
	If there are no unmarked edges connected to $v$ in both $G_1$ and
	$G_2$ then mark $v$ as closed.
	If there is no unmarked edge connected to $v$ in $G_1$ and there is
	some unmarked edge connected to $v$ in $G_2$ then we have found a
	difference as the number of tokens ``left'' in $v$ in $G_1$ is
	larger than in $G_2$. In this case, we have identified the
	difference between $G_1$ and $G_2$.
	Similarly for the other way around.

	The remaining case is when both in $G_1$ and $G_2$ there are
	unmarked edges connected to $v$.
	Let $e$ be the minimal unmarked edge connected to $v$ in $G_1$.
	If $e$ is not connected to $v$ in $G_2$ we are done.
	Indeed, the preset of $e$ either participate in $e$ in $G_1$ and
	not in $G_2$ or participate in a transition $L_E(e)$ in different
	ways in $G_1$ and $G_2$.

	Otherwise, $e$ is connected to $v$ both in $G_1$ and $G_2$.
	By its construction as a multiset of place histories, $e$ ``takes''
	the same number of tokens from $v$ in $G_1$ and $G_2$.
	As $e$ is unmarked, all the other nodes that $e$ takes tokens from
	have a sufficient number of unhandled tokens.
	Again, by $e$'s structure as a pair of a multiset and a transition,
	$e$ connects to exactly the same nodes in $G_1$ and $G_2$ in the same
	way.
	Reduce the marking of all predecessors of $e$ by the number of
	tokens taken by $e$ from them.
	If some of them are reduced to $0$ then they are closed.
	Mark $e$ as well.

	If there are no open nodes, then both $G_1$ and $G_2$ are finite
	and equivalent.
	Otherwise, continue handling open nodes by induction.
\end{appendixproof}

We note that by the proof of Theorem~\ref{thm:language glpo pn} all
the \lpos that disagree only on forced interleavings are refined by the
same \glpo.

\subsection{Choice vs Interleaving in CTSs}
We now proceed with CTS.
Here, a choice is either a situation where all the agents have exactly
the same history and at least one agent participates in a different
communication or communications on the same channel are ordered in a
different way.
Notice that as channels are global resources, the case that changing
the order of communications on a channel does not have side effects is
accidental. Indeed, such a change of order could have side effects and
constitutes a different choice.

We show that every two distinct g-computations of the
same CTS have a joint history of some agent that ``sees the
difference'' or a channel that transfers messages in a different order.
Difference for a history is either maximality in one and not the other
or extension by different communications in the two g-computations.

\begin{theoremrep}
	\label{thm:different glpos cts}
	Given a CTS $\mathcal{T}$ and two different \glpos
	$G_1,G_2\in\comp_g(N)$ then one of the following holds:
	\begin{enumerate}
 		\item For some agent $i$ there exists a history $h_i$ in both $G_1$
 		and $G_2$ such that either $h_i$ is maximal in $G_i$ and not
 		maximal $G_{3-i}$;
 		\item
 		For some agent $i$ there exists a history $h_i$ in both $G_1$ and
 		$G_2$ such that the edges	$e_1$ and $e_2$ such that $h_i \comreli{1}
 		e_1$ and	$h_i\comreli{2} e_2$ we have $L^1_E(e_1) \neq L^2_E(e_2)$;

 		\item
 		\label{difference 3 cts}
 		or; There is a pair of agents $i$ and $i'$ and histories $h_i$ and
 		$h_{i'}$ in both $G_1$ and $G_2$ such that the order between the
 		communications of $i$ and $i'$ is different in $G_1$ and $G_2$.
 	\end{enumerate}
\end{theoremrep}

As for PTI-nets, Theorem~\ref{thm:different glpos cts} is not true for
\lpos.
This does not hold as shown by the \lpos and \glpo of the CTS in
Figure~\ref{fig:cts}.
Recall, that this CTS has the same \lpos and \glpos depicted in
Figure~\ref{fig:petriex}(b).

\begin{appendixproof}
	As before, we define the \emph{depth} of elements in a partial
	order as their distance from a minimal element.
	Formally, the depth of the minimal element is $0$ and all the initial
	states (runs of length $1$) have depth of $1$.
	The depth of a non-minimal element $o$ is $\max_{o'\in \pre{o}}
	depth(o')+1$.

	As before, we order the elements in a \glpo by increasing depth.
	In addition, elements of the same depth are ordered so that edges
	appear before vertices and there is some arbitrary order between
	edges of the same depth and between vertices of the same depth.
	In this order, every element appears after all the elements that
	are smaller than it according to $\leq$.
	As before, every element has a finite depth and there is a finite
	number of elements in every depth.
	Hence, if we follow this order constitutes a linearization of the
	elements of the \glpo.

	We prove the theorem by induction according to the order mentioned
	above.
	As before, we are going to mark elements in the partial
	order as ``equivalent'' in both $G_1$ and $G_2$.
	The marking here is simpler (immediately closed marking).

	Consider the minimal element edges in $G_1$ and $G_2$ and their
	post-sets of runs of length 1 (depth 1).
	By definition, these correspond to the initial states of the
	different agents.
	It follows that they are the same.
	Mark all of them.

	Assume that we have marked up to a point in $G_1$ and $G_2$
	according to the induction order.
	We build the marking so that the maximal marked elements
	according to $\leq$ are	all nodes.
	Obviously, all maximal (according to $\leq$) marked elements are
	incomparable.
	It follows that we maintain the minimal unmarked element (in
	induction order) as an edge.
	Clearly, this is true for the marking of the minimal nodes.

	Consider the set of unmarked edges in $G_1$ and $G_2$.
	If both are empty, then $G_1$ and $G_2$ are the same.
	Suppose that the set of unmarked edges in (wlog) $G_1$ is empty and
	$G_2$ is not empty.
	Consider the sender participating in the communication of the first
	unmarked edge in $G_2$.
	It must be the case that we have found an agent $i$ and a history
	$h_i$ that is maximal in $G_1$ and not maximal in $G_2$.
	The remaining case is that both $G_1$ and $G_2$ have
	unmarked edges.

	Consider the \glpo $G_1$.
	Let $e$ be the minimal unmarked edge in $G_1$ according to the
	induction order.
	Let $h_1,\ldots, h_n$ be $\pre{e}$ in $G_1$ with $h_1$ being the
	sender.
	As all elements of smaller depth than $e$ have been marked, it
	follows that $h_1,\ldots, h_m$ have been marked and that they
	appear also in $G_2$.

	Consider a history $h_i\in \pre{e}$.
	If $h_i$ is maximal in $G_2$ we are done.
	Otherwise, let $e_i$ be the edge such that $h_i\succi{2} e_i$.
	If $L^1_E(e)\neq L^2_E(e_i)$ we are done as $h_i$ does something
	different in $G_1$ and $G_2$.
	The same holds for every $j \in \{1,\ldots, m\}$.
	Hence, for every $j$ we have $e_j$ exists and $L_E(e_j)=L(e)$.

	Suppose that $G_1$ and $G_2$ are different here.
	This can only happen if there are at least two agents $j$ and $j'$
	for which $e_j$ and $e_{j'}$ are distinct edges labeled by the same
	communication.
	In particular, $n\geq 2$ and the agents in histories $h_i$ for
	$i>1$ are listening to channel $ch(L_E(e))$.

	However, for $e_j$ and $e_{j'}$ each, there is a unique sender.
	If $h_1$ is not sending in $G_2$ then $h_1$ does something
	different in $G_1$ and $G_2$ and we are done.
	Wlog, assume that $h_1$ is the sender of $e_j$.
	Consider the following options.
	\begin{itemize}
		\item
		Suppose that one of the agents $h_i$ for $i>1$ is the sender of
		$e_{j'}$.
		Then, $h_i$ is a history that receives in $G_1$ and sends in
		$G_2$.
		Thus, $h_i$ does something different in $G_1$ and $G_2$.
		\item
		Suppose that there exists an additional agent $k$ and a history
		$h_k$ such that $h_k$ is the sender for $e_{j'}$.
		In order \emph{not} to find a difference between $G_1$ and $G_2$,
		it must be the case that $h_k$ is a sender of $e_{j'}$ also in
		$G_1$ and the set of agents that participate in $e_{j}$ and $e_{j'}$
		\emph{together} is the same and they have the same roles.
		That is, every agent that is a receiver in $G_1$ is a receiver in
		$G_2$ and vice versa.
		However, as we assumed that $G_1$ and $G_2$ are different, there
		are again two options:
		\begin{itemize}
			\item Either the order between $e_j$ and $e_{j'}$ in $G_1$ and
			$G_2$ is reversed.
			This matches the difference~\ref{difference 3 cts}, where the
			senders are the agents witnessing the difference.
			\item Or the order between $e_j$ and $e_{j'}$ is the same in both
			$G_1$ and $G_2$.
			Then, the matching between senders and receivers in $G_1$ and
			$G_2$ to $e_j$ and $e_{j'}$ is different.
			Consider a receiver that moved from listening to (wlog) $e_j$ to
			$e_{j'}$.
			It follows that this agent participates in an early communication
			in $G_1$ ($e_j$) and a later communication in $G_2$ ($e_{j'}$).
			This receiving agent and the sender of $e_{j}$ see a different
			order of the communication they participate in (from equal to one
			before the other).
		\end{itemize}
	\end{itemize}

	By induction, unless this process terminates prematurely by finding
	a difference, it will visit all of $G_1$ and $G_2$ and show that
	they are, in fact, equivalent.
\end{appendixproof}

We note that by the proof of Theorem~\ref{thm:language glpo cts} all
the \lpos that disagree only on forced interleavings are refined by the
same \glpo.
\section{Concluding Remarks}\label{sec:conc}
In this paper, we laid down the basis to reason about reconfiguration in concurrent systems from a global perspective. We showed how to isolate forced interleaving decisions of the system due to reconfiguration, and other decisions due to standard concurrent execution of independent events. To test our results, we considered
PTI-nets~\cite{FlynnA73,pnarc} and
CTS~\cite{alrahman2021modelling,DBLP:conf/atal/AlrahmanPP20} which cover a wide range of
interaction capabilities alongside reconfiguration from two different
schools of concurrency.
We proposed, for both, a partial order semantics, named \lpo, of computations under reconfiguration. An \lpo extends occurrence nets~\cite{Vogler02a} with event-to-event connections that allows to refer to reconfiguration points. Moreover, to fully characterise reconfiguration in a single structure, we proposed a glued \lpo semantics, named \glpo. The latter is able to fully isolate scheduling decisions due to reconfiguration from the ones due to standard concurrency.  We show that any $\mathsf{LPO}$ computation is only a refinement of some
g-$\mathsf{LPO}$ of the same system. Finally, we prove important results on
g-$\mathsf{LPO}$ with respect to reconfiguration and
nondeterminism.

For future work, we would like to exploit \glpo semantics to verify
properties about reconfiguration and interaction in general.
Namely, we would like to define a specification logic that considers
\glpo computations as the underlying structure rather than the standard
linear computations.
Clearly, logics over linear structures easily distinguish different
interleavings of the same \lpo.
However, different linearizations of the same \lpo are either all
computations of a system or none of them is.
Similarly, a logic defined over \lpos would easily distinguish
different schedules that relate to reconfiguration.
Again, these different \lpos are either all computations of a system or
none of them is.
By considering \glpos as the underlying structure we can create
specifications that do not distinguish between different schedules that
correspond to the same choices of the system.
Our view is that such a specification language that incorporates
elements of Strategy logic~\cite{ChatterjeeHP10} and
\ltal~\cite{DBLP:conf/atal/AlrahmanPP20} would not only allow us
to reason about interaction and reconfiguration, but also to reason
about the local views of agents as well as their combined behaviour.

\paragraph*{\bf Related works}
The prevalent approach to semantics of reconfigurable interactions is
based on linear order semantics (cf.
Pi-calculus~\cite{milner1992calculus,ene1999expressiveness}, Mobile
Ambients~\cite{luca1}, Applied
Pi-calculus~\cite{appliedpi}, Psi-calculus~\cite{broadcastpsi,psi},
concurrent constraint programming~\cite{ccp,ccpp}, fusion
calculus~\cite{fusion}, the \abc calculus~\cite{info,scp},
\rcp~\cite{DBLP:conf/atal/AlrahmanPP20} etc.).  This semantics cannot
distinguish the different choices of the system from a global
perspective, and thus does not facilitate reasoning about
reconfiguration from an external observer's point of view. It also hides
information about interactions and possible interdependence among
events. In fact, linear order semantics ignores the possible
concurrency of events, which can be important e.g. for judging the
temporal efficiency of the system~\cite{Vogler02a}.
However, it still provides a correct abstraction of the system behaviour, while hiding such details.

Partial order semantics (cf. \emph{Process semantics} of Petri nets~\cite{PetriR08,meseguer1992semantics,Vogler02a}
and \emph{Mazurkiewicz traces} of Zielonka
automata~\cite{Zielonka87,GenestGMW10,KrishnaM13}), on the other hand, is able to refer to the interaction and event dependencies, but does not deal very well with reconfiguration. This is because the latter formalisms have fixed interaction structures, and thus the interdependence of events is defined structurally.
Reconfiguration, on the other hand, enforces reordering of events dynamically in non-trivial ways, and thus makes defining correct partial order semantics very challenging.
As shown in~\cite{JanickiKKM21}, some aspects of concurrency are almost
impossible to tackle in both linear-order and partial-order
causality-based models, and one of them is PTI-nets~\cite{FlynnA73}.
In fact,
reconfiguration increases the expressive power of the formalism, e.g.,
adding inhibitor arcs to Petri nets makes them Turing
Powerful~\cite{agerwala1974complete}. However, this expressive power
does not come without expenses. In fact, it prevents most analysis
techniques for standard Petri nets~\cite{pnarc}.

To the best of our knowledge, the closest to our \lpo semantics is
Relational Structures~\cite{JanickiKKM21}.
In order to capture inhibition they add an additional ``not later
than'' relation to partial orders.
Much like our \lpos, this allows to represent the different forced
interleavings \emph{separately}.
The emphasis in \cite{JanickiKKM21} is on providing a general semantic
framework for concurrent systems.
Thus, relational structures handle issues like priority and error
recovery, which we do not handle.
However, relational structures are not concerted directly with
separation of choice from
interleaving as we are.
So the two works serve different purposes and it would be interesting
to investigate mutual extensions.

%
%
%
%
%
%
%
%

%
%
%
\bibliographystyle{splncs04}
\bibliography{references}

\appendix

\end{document}